\documentclass[conference,letter,10pt]{llncs}
\usepackage{times}
\usepackage{amssymb}
\usepackage{times,amsmath,epsfig,subfigure}
\usepackage{algorithm,algorithmic}
\usepackage{appendix}
\usepackage{soul}
\usepackage{color}
\floatname{algorithm}{Algorithm}
\newcommand{\comment}[1]{}
\providecommand{\keywords}[1]{\textbf{\textit{Keywords:  }}#1}
\begin{document}

\title{Securely Solving the Distributed Graph Coloring Problem\thanks{The short version of this paper appeared in the proceedings of 2011 IEEE International Conference on Privacy, Security, Risk and Trust \cite{HongTabu}.}}

\author{Yuan Hong\textsuperscript{$\dagger$} \and Jaideep Vaidya\textsuperscript{$\ddagger$} \and Haibing Lu\textsuperscript{$\S$}}
\institute{\textsuperscript{$\dagger$}Illinois Institute of Technology\\ \textsuperscript{$\ddagger$}Rutgers University\\ \textsuperscript{$\S$}Santa Clara University \\
	\email{\textsuperscript{$\dagger$}yuan.hong@iit.edu\\
		\textsuperscript{$\ddagger$}jsvaidya@business.rutgers.edu\\
		\textsuperscript{$\S$}hlu@scu.edu}
}

\maketitle
\begin{abstract}
Combinatorial optimization is a fundamental problem found in many fields. In many real life situations, the constraints and the objective function forming the optimization problem are naturally distributed amongst different sites in some fashion. A typical approach for solving such problem is to collect all of this information together and centrally solve the problem. However, this requires all parties to completely share their information, which may lead to serious privacy issues. Thus, it is desirable to propose a privacy preserving technique that can securely solve specific combinatorial optimization problems. A further complicating factor is that combinatorial optimization problems are typically NP-hard, requiring approximation algorithms or heuristics to provide a practical solution. In this paper, we focus on a very well-known hard problem -- the distributed graph coloring problem, which has been utilized to model many real world problems in scheduling and resource allocation. We propose efficient protocols to securely solve such fundamental problem. We analyze the security of our approach and experimentally demonstrate the effectiveness of our approach.

\end{abstract}


\section{Introduction}
\label{sec:intro}
Optimization is a fundamental problem found in many fields, and it has many applications in operational research, artificial intelligence, machine learning and mathematics. In many real life situations, the constraints and the objective function forming the optimization problem are naturally distributed amongst different sites in some fashion. A typical approach for solving such problem is to collect all of this information together and centrally solve the problem. For example in supply chain management, since the delivery trucks are not always fully loaded (i.e. empty $25\%$ of the time), a collaborative logistics service (Nistevo.com) can be utilized to merge loads from different companies bound to the same destination. Huge savings were realized in such collaborative transportation (freight costs were cut by 15\%, for an annual savings of \$2 million\cite{Turban-LoL}), which can be modeled as a typical distributed linear programming problem.

However, this requires all parties to completely share their information, which may lead to serious privacy issues simply because each party's portion of the optimization problem typically refers to its private information. For instance, in constrained optimization, the private constraints often describe the limitations or preferences of a party, and each party's share of the global solution represents its
private output in the global best decision or assignments. Furthermore, even without privacy constraints, solving an optimization problem typically is computationally very expensive. 

Thus, an open question is whether we can solve such problems both securely and efficiently? This challenging issue has been investigated only in the context of a few optimization problems. Specifically, Vaidya \cite{Vaidya09SAC} securely solved the collaborative linear programming problem, assuming that one party holds its private objective function while the other party holds its private constraints. Sakuma et al. \cite{SakumaK07} proposed a genetic algorithm for securely solving two-party distributed traveler salesman problem (TSP) \cite{tsp18}, assuming that one party holds the cost vector while the other party holds the exact cities that should be visited. Atallah et al. \cite{sec_supplychain} proposed protocols for secure supply chain management. However, none of these can be used to solve general problems such as distributed scheduling or network resource allocation \cite{Marx04graphcoloring}. In this paper, we focus on this class of problems. For instance, consider a simple multi-party job scheduling problem:

\begin{example}
\textit{Alice, Bob and Carol have jobs $\{1,2,3\}$,$\{4,5\}$ and $\{6,7\}$ respectively; as depicted in Fig. \ref{fig:vertex} (edges represent conflicts), some jobs cannot be assigned into the same time slot due to resource conflicts (i.e. sharing the same machine); Since each job is held by a specific party,   they do not want to reveal their set of jobs and corresponding assigned time slots to each other. Given this, is it possible to securely assign the $7$ jobs into $k$ time slots?}
\end{example}

\begin{figure}[h]
\centering \subfigure[Partitioned
Graph]{\centering\includegraphics[angle=0,
width=0.27\linewidth]{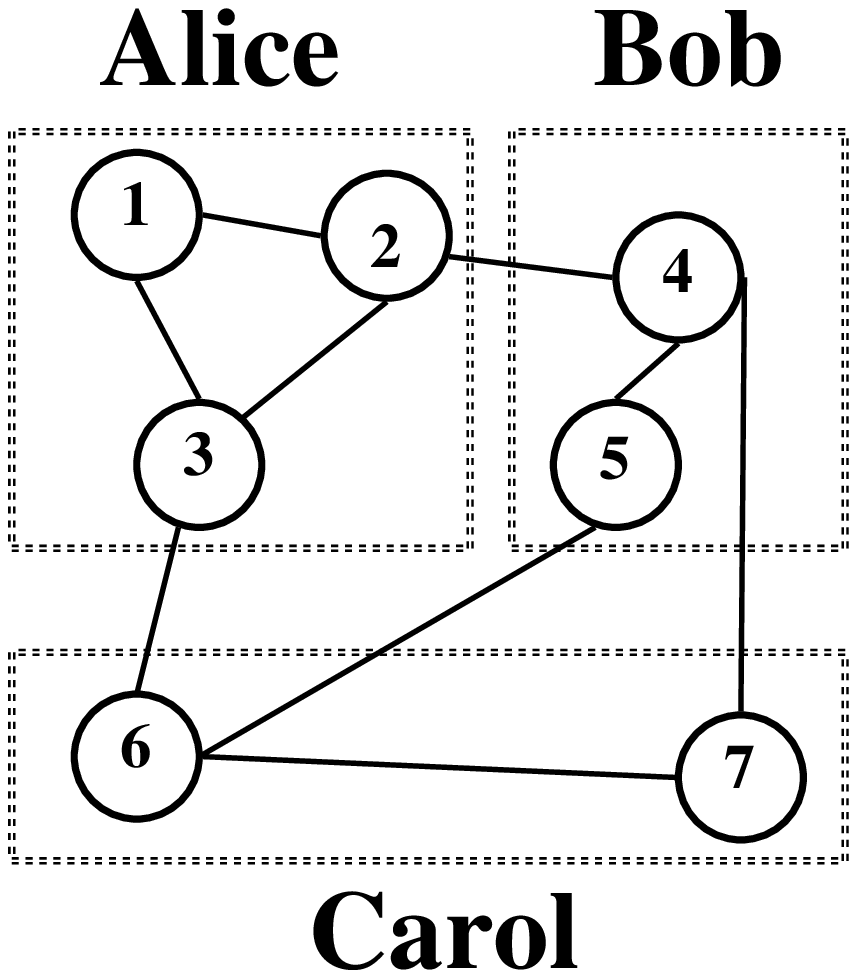} \label{fig:vertex}
}\hspace{0.3in}\subfigure[Partitioned Adjacency
Matrix]{\centering\includegraphics[angle=0,
width=0.44\linewidth]{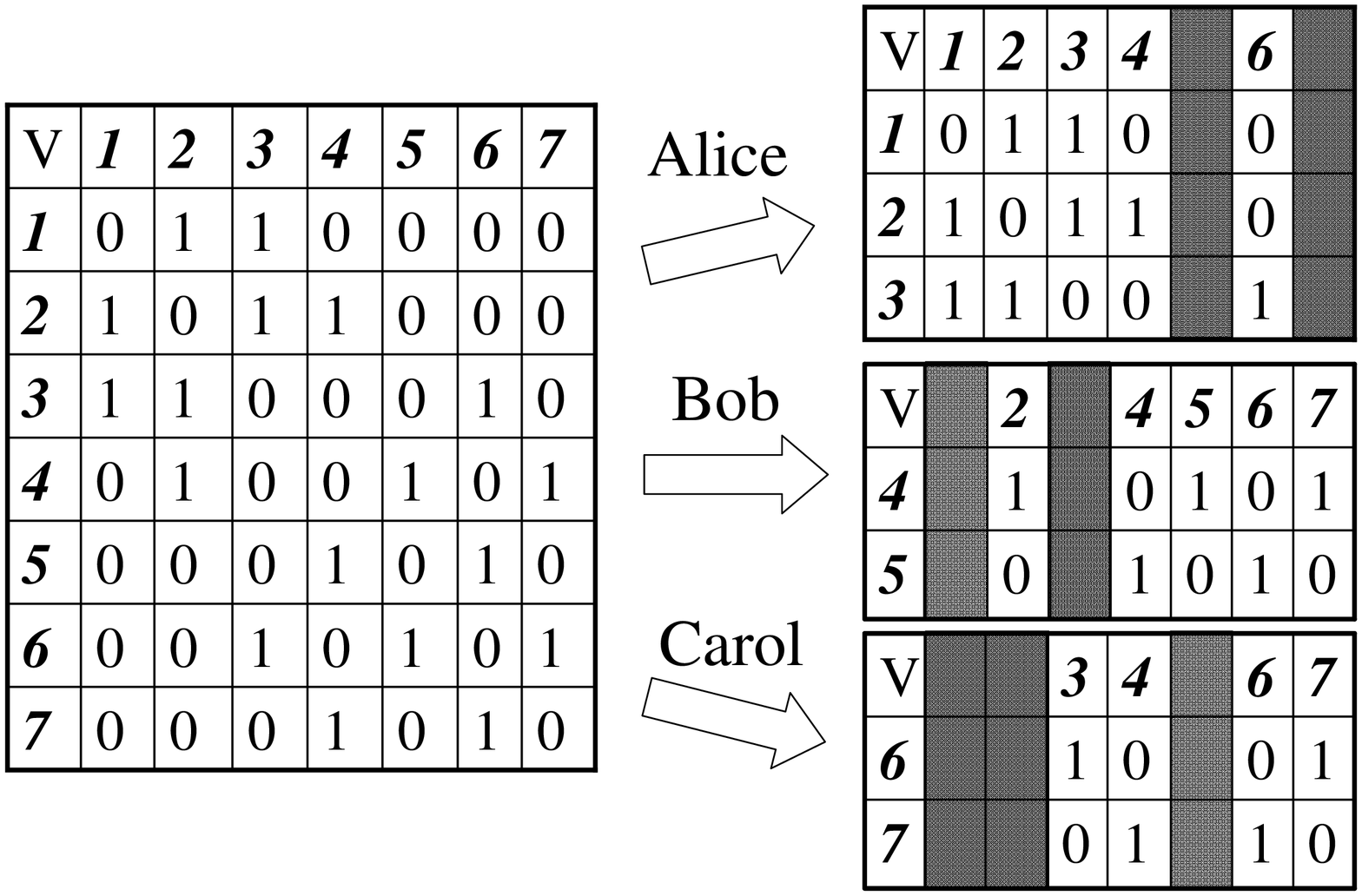} \label{fig:adjv} } \caption{
Distributed Graph Coloring Problem} \label{fig:part}
\end{figure}

Intuitively, the above scheduling problem can be modeled as a distributed (among $3$ parties) k-coloring problem (the decision problem of graph coloring). Specifically, given $k$ different colors ($k$ time slots), the undirected graph $G=(V,E)$ represents the jobs and their conflicts (as shown in Fig. \ref{fig:vertex}) -- a vertex represents a job, and an edge between two vertices represents a conflict between the corresponding jobs. Now, each vertex should be assigned one out of $k$ colors (representing the time assignment for a job), with no two vertices connected by an edge having the same color (no two jobs having a conflict should be scheduled into the same time slot). Thus, a feasible k-coloring solution represents a feasible scheduling solution.

Note that $G$ is partitioned such that each party owns a subgraph (i.e, holds some vertices and the edges related to its own vertices). As shown in Fig. \ref{fig:vertex}, Alice, Bob and Carol hold $3$, $2$ and $2$ jobs, respectively. The adjacency matrix of the graph is thus partitioned as shown in Fig. \ref{fig:adjv}. \comment{In Example $2$, $G$ is partitioned in terms of edges among $m$ parties (we denote it as \textit{edge-partitioned} DisGC later): each party holds a subset of edges in $G$ that indicates the conflicts among a subset of vertices (All the vertices are known to all parties, but the edges and the corresponding vertices are not revealed). As shown in Fig. \ref{fig:edge} and \ref{fig:adje}, Alice, Bob and Carol holds $4, 3$ and $3$ edges among all $7$ vertices, respectively. However, the adjacency matrix of the graph is arbitrarily partitioned (also symmetric) among three parties.} When formulating and solving this problem, Alice, Bob and Carol do not want to reveal their partition in the graph and the colors of their vertices to each other. Hence, the primary goal of this paper is to propose a privacy-preserving solver for the distributed graph coloring problems.

Furthermore, in many combinatorial optimization problems, exhaustive search is not feasible due to the NP-hard nature of them. Since graph coloring problem has been proven to be NP-hard (its decision problem, $k$-coloring problem is NP-complete), securely solving the graph coloring problem should be extremely challenging due to the computational complexity. Several meta-heuristics such as tabu search have been proposed to solve the centralized graph coloring problem by improving the performance of local search with memory structures \cite{tabu1}\cite{tabu2}\cite{Hertz}. Even if meta-heuristics cannot guarantee an optimal results, they are generally considered to be quite effective in solving combinatorial optimization problems and obtain near-optimal outputs in an acceptable running time. Therefore, in this paper, we propose a secure and efficient algorithm to enable distributed tabu search. 

The rest of this paper is organized as follows. Section \ref{sec:pro} formally defines the problem and reviews the tabu search method for graph coloring. Section \ref{sec:vertex} introduces a secure and efficient protocol to solve distributed graph coloring problem, assuming that all parties are honest-but-curious (obey the protocol). Section \ref{sec:security} provides a formal security analysis and complexity analysis. Section \ref{sec:exp} experimentally validates the performance of our approach. Section \ref{sec:con} concludes the paper and discusses some directions for future work.

\section{Problem Definition}
\label{sec:pro}
In graph theory, the \textbf{k-coloring} problem represents the following NP-complete decision problem\cite{bookco} -- does graph $G$ have a proper vertex coloring with $k$ colors.
The (NP-hard) optimization version of this problem simply seeks the minimum $k$ that can properly color the given graph. The minimum $k$ is known as the chromatic number of the graph. Since, seeking the chromatic number can be reduced to a set of $k$-coloring problems, we only consider the $k$-coloring problem on distributed graphs in this paper. We first formally define the graph coloring (decision problem), review tabu search for graph coloring, and then introduce our solution approach.

\subsection{Graph Coloring Problem Formulation with Scalar Products}
\label{sec:gc_pro}

When seeking the chromatic number $k$ or solving the $k$-coloring problem, any solution requires that the number of adjacent vertices having the same color should be %minimized to
$0$. We thus denote any two adjacent vertices that have the same color as a ``conflict'' in
a colored graph. Thus, given $k$ different colors, if we let $\mu$ represent the total number of conflicts in a colored $G$, we thus have: if $\min(\mu)=0$, $G$ is $k$-colorable.

Given $k$ colors, we can represent the color of any vertex in the graph using a $k$-dimensional boolean vector $x$ (domain $\Phi(k)$: all the $k$-dimensional boolean vectors include only one ``1'' and $k-1$ ``0''s). If the $ith$ number in the boolean vector is $1$ (all the remaining numbers should be $0$), the current vertex is colored with the $ith$ color ($i\in[1,k]$). Thus, we let $x=\{x_i,\forall
v_i\in G\}$ be a coloring solution in a $k$-coloring problem where $x_i$ represents the coloring boolean vector for vertex $v_i$. In addition, given the coloring boolean vectors of two adjacent vertices $x_i,x_j\in x$ ($i\ne j$ and $v_i,v_j\in G$), we can represent the color conflict of $v_i$ and $v_j$ in solution $x$ as the \textbf{scalar product} of $x_i$ and $x_j$:

\begin{equation}
\mu_{ij}(x) = \left\{
\begin{array}{rl}
0 & \text{   if  } x_i\cdot x_j = 0 \text{ ($v_i$ and $v_j$ are not conflicted)}\\
1 & \text{   if  } x_i\cdot x_j = 1 \text{ ($v_i$ and $v_j$ are
conflicted)}
\end{array} \right.
\end{equation}

where $\mu(x)=\sum_{\forall e_{ij}\in G}\mu_{ij}(x)$. We thus formulate the $k$-coloring problem as $\{\min: \sum_{\forall e_{ij}\in G} x_i\cdot x_j$, $s.t. \forall v_i,v_j\in G, x_i,x_j\in \Phi(k)\}$. If and only if the optimal value is $0$, the problem is $k$-colorable. When seeking the chromatic number $k$, we can solve the above problem using different $k$ and choose the minimum colorable $k$.

\subsection{Tabu Search for Solving Graph Coloring Problems}
\label{sec:tabucol}

Tabu search uses a memory structure to enhance the performance of a local search method. Once a potential solution has been determined, it is marked as ``taboo''.  So the algorithm does not visit that
possibility repeatedly. Tabu search has been successfully used for graph coloring. Specifically, given a graph $G=(V,E)$ and $k$ colors, a popular tabu search algorithm for graph coloring -- Tabucol \cite{Hertz} (summarized using our notation) is described as follows:

\begin{enumerate}

  \item initialize a solution $x=(\forall v_i\in G, x_i\in \phi(k))$ by randomly selecting colors.
  \item iteratively generate neighbors $x'$ of $x$ until $\mu(x')<\mu(x)$ by changing the color of an endpoint of an arbitrary conflicted edge $v_i:x_i\rightarrow x_i'$
  (if no $x'$ is found with $\mu(x')<\mu(x)$ in a specified number of iterations, find a best $x'$ in current set of neighborhoods where $(v_i,x_i')$ is not in the tabu list, and let $x'$ be the next solution).
  \item if ($v_i,x_i'$) is not in the tabu list, make current $x'$ (best so far) as the next solution $x=x'$ and add the pair $(v_i,x_i)$ into the tabu list; If $(v_i,x_i')$ is in the tabu list and $\mu(x')$ is still smaller than $\mu(x)$, make $x'$ as the next solution anyway; Otherwise, discard $x'$.

  \item repeat step 2,3 until $\mu(x')=0$ or the maximum number of iterations is reached.
\end{enumerate}

\subsection{Securely Solving Distributed Graph Coloring Problems}

Recall that Example $1$ can be modeled by a distributed graph coloring problem on a graph $G$ whose adjacency matrix is partitioned among $m$ parties (each vertex belongs to only one
party). We thus define it as follows:

\begin{definition}[Distributed Graph Coloring (DisGC)]
Given $k$ colors and a partitioned graph $G=(V,E)$ where each vertex $v_i\in V$ belongs to only one party, does a $k$-coloring of $G$ exist?

\end{definition}

In the distributed graph, we define two different types of edges: 1. \textbf{internal edge: } any edge $e_{ij}$ whose endpoints ($v_i$ and $v_j$) are owned by the same party; 2. \textbf{external edge: } any
edge $e_{ij}$ whose endpoints ($v_i$ and $v_j$) belong to two different parties. Similarly, we define two categories of vertices: 1. \textbf{inner vertex}: a vertex that does not have any external edge
(e.g.,. $v_1$ in Fig. \ref{fig:vertex}) 2. \textbf{border vertex}: a vertex that has at least one external edge (e.g., $v_2$ in Fig. \ref{fig:vertex}).

When $m$ different parties want to collaboratively solve the DisGC problem, each party does not want to reveal its partition in $G$ and the colors of its vertices to others. For instance, suppose we color the
graph in Fig. \ref{fig:vertex} using $k=3$ colors, the parties should only learn the relevant information. For example, Fig. \ref{fig:colored} gives a $3$-colored global solution, corresponding to which Alice, Bob and Carol would only know the information shown in Fig. \ref{fig:gca}, \ref{fig:gcb} and \ref{fig:gcc}, respectively. Thus, each party only sees the colors assigned to the vertices in their local graph and cannot even see the color of the border vertices of other parties that share an external edge with this party (the vertex existence is known anyway). For example, in Alice's view of Bob's graph, only vertex 4 which is directly connected to vertex 2 in Alice's graph is visible to Alice, though its color is not known to Alice.

\begin{figure}[h]
\centering \subfigure[3-Colored]{\centering\includegraphics[angle=0,
width=0.19\linewidth]{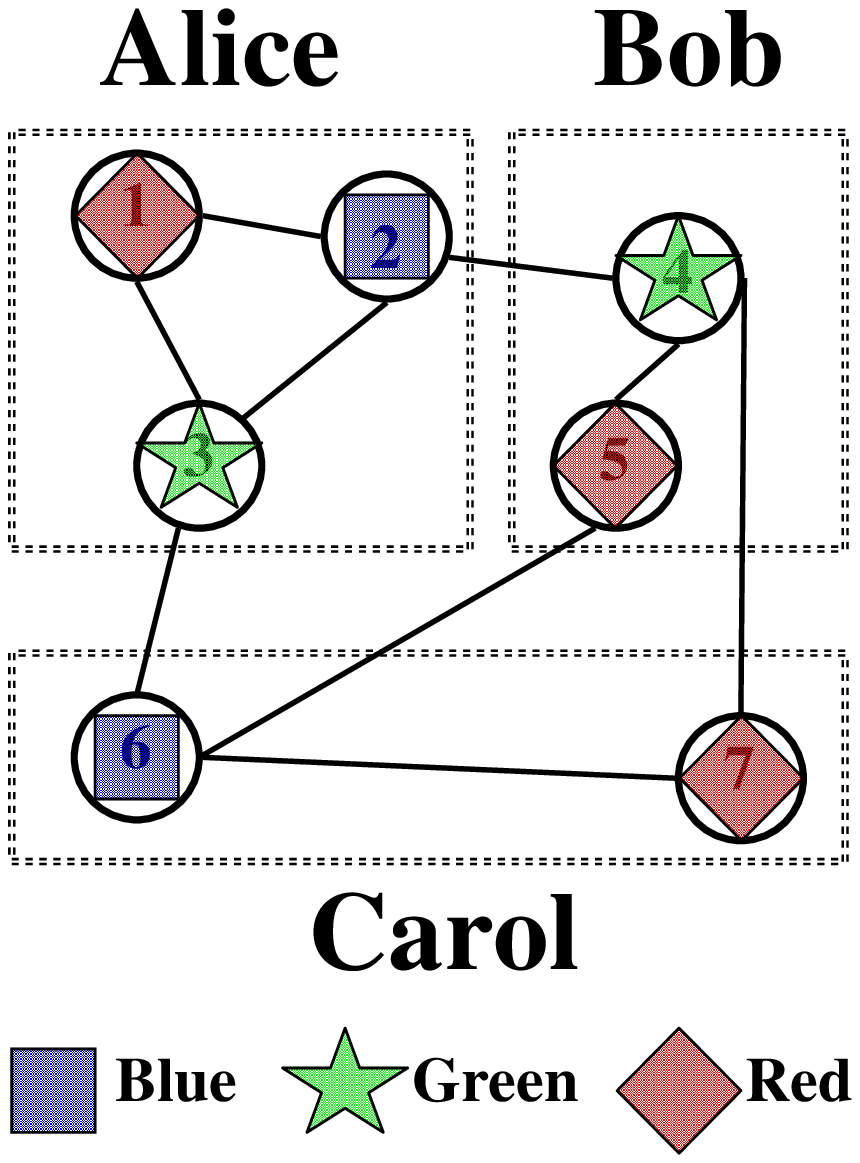} \label{fig:colored} } \hspace{0.05in}
\subfigure[Alice]{\centering\includegraphics[angle=0,
width=0.19\linewidth]{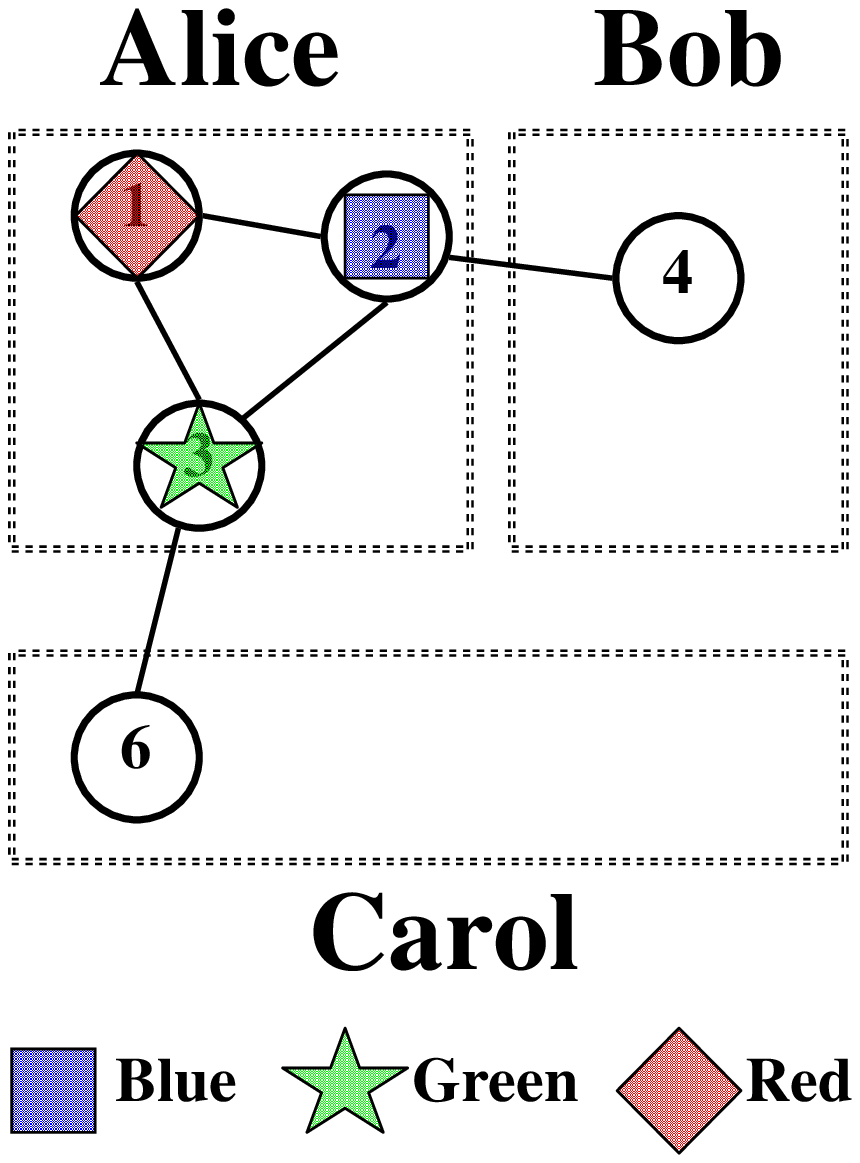} \label{fig:gca} }
\hspace{0.05in}\subfigure[Bob]{\centering\includegraphics[angle=0,
width=0.19\linewidth]{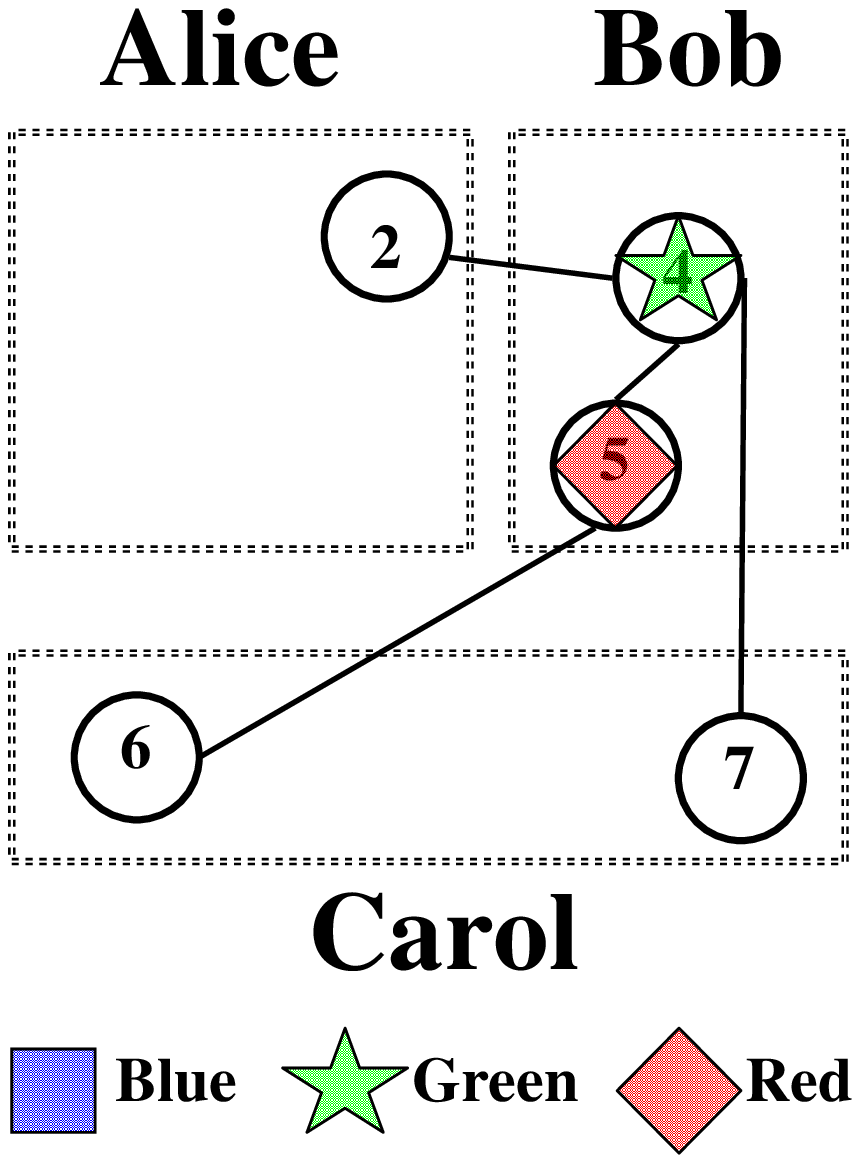} \label{fig:gcb} }
\hspace{0.05in}\subfigure[Carol]{\centering\includegraphics[angle=0,
width=0.19\linewidth]{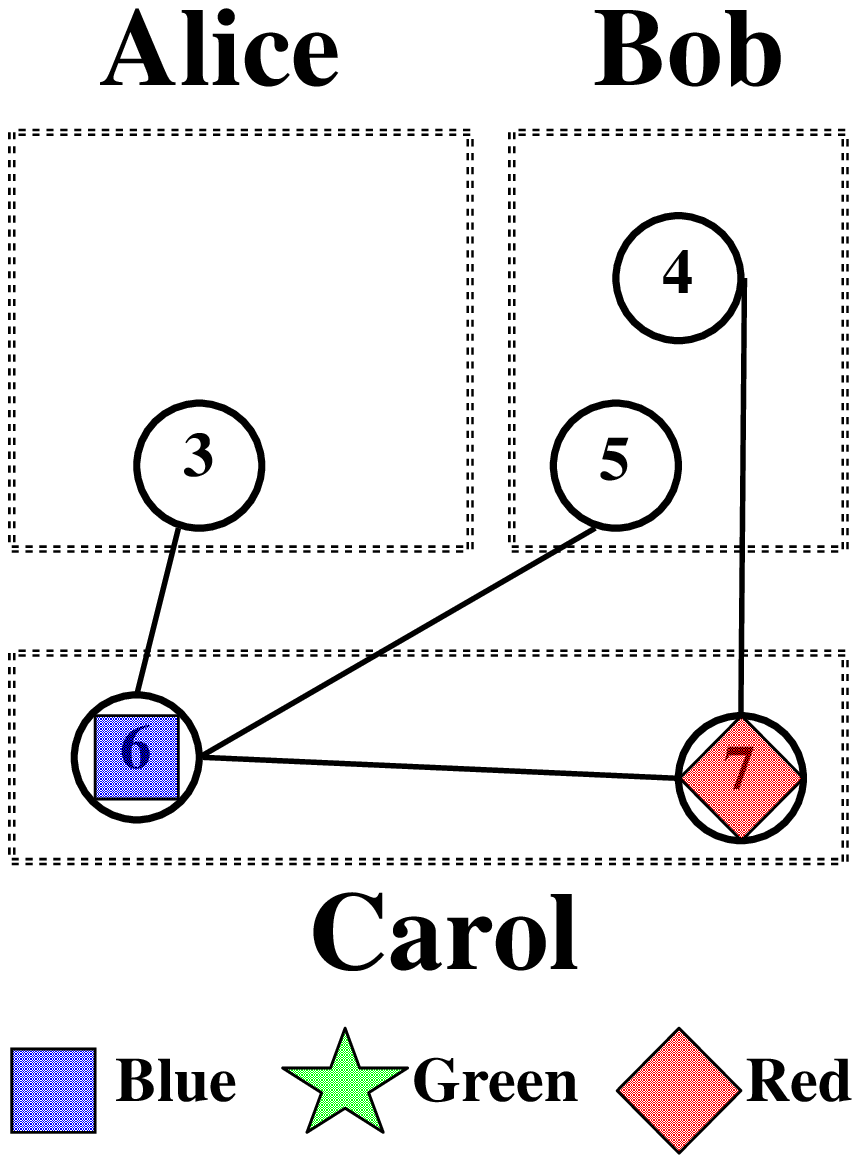} \label{fig:gcc} }
\vspace{-0.1in}\caption[Example of Distributed Combinatorial
Optimization] {A Secure Solution to the DisGC Problems}
\label{fig:gc}
\end{figure}

Thus, we have two privacy requirements when solving the DisGC problem: first, protect each party's subset of the graph structure including its vertices and the internal edges (the external edges and their endpoints are inevitably known to the two parties sharing them); second, protect the color assignment for each party's  vertices (including the endpoints of the external edges) in every iteration. A direct method to securely solve the DisGC problem is by using a trusted third party to collect partitioned data from all parties, centrally solving the problem with Tabucol algorithm \cite{Hertz} and distributing the color of every vertex to the vertex owner. However, if no trusted third party can be found, a secure protocol is necessary. Therefore, the goal of this paper is to present a privacy-preserving tabu search method for solving DisGC without a trusted third party, assuming that all parties are honest-but-curious. The honest-but-curious assumption is realistic, and is commonly assumed in the literature, and the solution can be generalized to fully malicious adversaries (though at added cost).

\section{Privacy Preserving Tabu Search}
\label{sec:vertex}
As mentioned in Section \ref{sec:tabucol}, a typical tabu search algorithm for graph coloring problem includes numerous repeated computations such as: computing the total number of conflicts for a solution, searching better neighborhoods for a specific solution, verifying a new vertex color pair with all the taboos and updating the tabu list if necessary. We now look at how all of this can be securely done, starting with some of the fundamental cryptographic building blocks that will be used to provide a solution.

\subsection{Fundamental Cryptographic Building Blocks}

\subsubsection{Secure Scalar Product for Secure Conflict Computation}

For any potential solution (a colored graph) to the DisGC problem, the total number of conflicts corresponding to that solution must be computed to figure out whether it improves the objective. As introduced in Section \ref{sec:gc_pro}, the total number of conflicts in a solution $x$ is denoted by $\mu(x)$. $\mu(x)$ can be represented as the sum of all the scalar products of every two adjacent vertices' coloring vectors. To securely compute the scalar product of two endpoints on every external edge, we use Goethals et al.'s approach \cite{Goethals-scalprod} which is based on a public-key homomorphic cryptosystem. This is both simple and provably secure. The problem is defined as follows: $P_a$ has a $n$-dimensional vector $\vec{X}$ while $P_b$ has a $n$-dimensional vector $\vec{Y}$. At the end of the protocol, $P_a$ should get vector vector $\vec{X} \cdot \vec{Y} +r_b$ where $r_b$ is a random number chosen from a uniform distribution and is known only to $P_b$. The key idea behind the protocol is to use a homomorphic encryption system such as the Paillier cryptosystem \cite{Paillier99}. A homomorphic cryptosystem is a semantically-secure public-key encryption system which has the additional property that given any two encryptions $E(A)$ and $E(B)$, there exists an encryption $E(A*B)$ such that $E(A)*E(B) = E(A*B)$, where $*$ is either addition or multiplication (in some Abelian group). The cryptosystems mentioned above are additively homomorphic (thus the operation $*$ denotes addition). Using such a system, it is quite simple to create a scalar product protocol. If $P_a$ encrypts its vector $\vec{X}$ and sends $Enc_{pk}(\vec{X})$ with the public key $pk$ to $P_b$, $P_b$ can use the additive homomorphic property to compute the encrypted sum of scalar product and its random number: $Enc_{pk}(\vec{X}\cdot\vec{Y}+r_b)$. Finally, $P_b$ sends $Enc_{pk}(\vec{X}\cdot\vec{Y}+r_b)$ back to $P_a$. Thus, $P_a$ decrypts and obtain $\vec{X}\cdot\vec{Y}+r_b$ ($r_b$ and $\vec{X}\cdot\vec{Y}$ are unknown to $P_a$). Since $P_b$ cannot decrypt $Enc_{pk}(\vec{X}\cdot\vec{Y}+r_b)$ with a public key, it cannot learn the scalar product either.

Therefore, we can adopt the above cryptographic building block to securely compute the total number of conflicting edges $\mu(x)$ in solution $x$ (Secure Conflict Computation). Note: while computing
$\mu(x)$, it is not necessary to securely compute the number of each party's conflicting internal edges in $x$. Alternatively, we can let each party independently count them, and sum the number with its shares in secure scalar product computation. The detailed steps are shown in Algorithm \ref{algo:svc}.

\begin{algorithm}[!h]
\begin{algorithmic}[1]

\renewcommand{\algorithmicrequire}{\textbf{Input:}}
\renewcommand{\algorithmicensure}{\textbf{Output:}}

\REQUIRE $G=(V,E)$ is colored as $x\in \Phi(k)$ where the color of
$\forall v_i\in G$ is denoted as $x_i\in x$.

\ENSURE each party $P_a$'s share $\mu_a(x)$ in $\mu(x)$ where
$\mu(x)=\sum_{\forall a}\mu_a(x)$;

\FOR{each party $P_a$}

\STATE count the total number of its conflicting internal edges as
$\mu_a(x)$;

\ENDFOR

\FORALL{external edges $e_{ij}\in G$}

\STATE $****$ Suppose that $v_i$ and $v_j$ belong to $P_a$ and
$P_b$, respectively. $****$

\STATE W.o.l.g. $P_a$ creates a public and private key pair $(pk,
sk)$;

\STATE $P_a$ encrypts its color vector $x_i$ as $x_i'=Enc_{pk}(x_i)$
and sends $x_i'$ and $pk$ to $P_b$;

\STATE $P_b$ generates a random integer $r_{ij}$, encrypts the
scalar product as: $x_i''=Enc_{pk}(x_i\cdot x_j)*Enc_{pk}(r_{ij})$
with the public key $pk$, and sends $x_i''$ back to $P_a$;

\STATE $P_a$ decrypts $x_i''$ with its private key $sk$ and obtain
$s_{ij}=x_i\cdot x_j+r_{ij}$;

\ENDFOR

\FOR{each party $P_a$}

\STATE $\mu_a(x)\leftarrow \mu_a(x)+\sum \forall s_{ij}$ (received
from other parties)$-\sum\forall r_{ij}$ (generated by $P_a$);

\ENDFOR

\end{algorithmic}
\caption{Secure Conflict Computation}\label{algo:svc}
\end{algorithm}

\begin{figure}[h]
\centering \subfigure[Conflicts in
$G$]{\centering\includegraphics[angle=0,
width=0.28\linewidth]{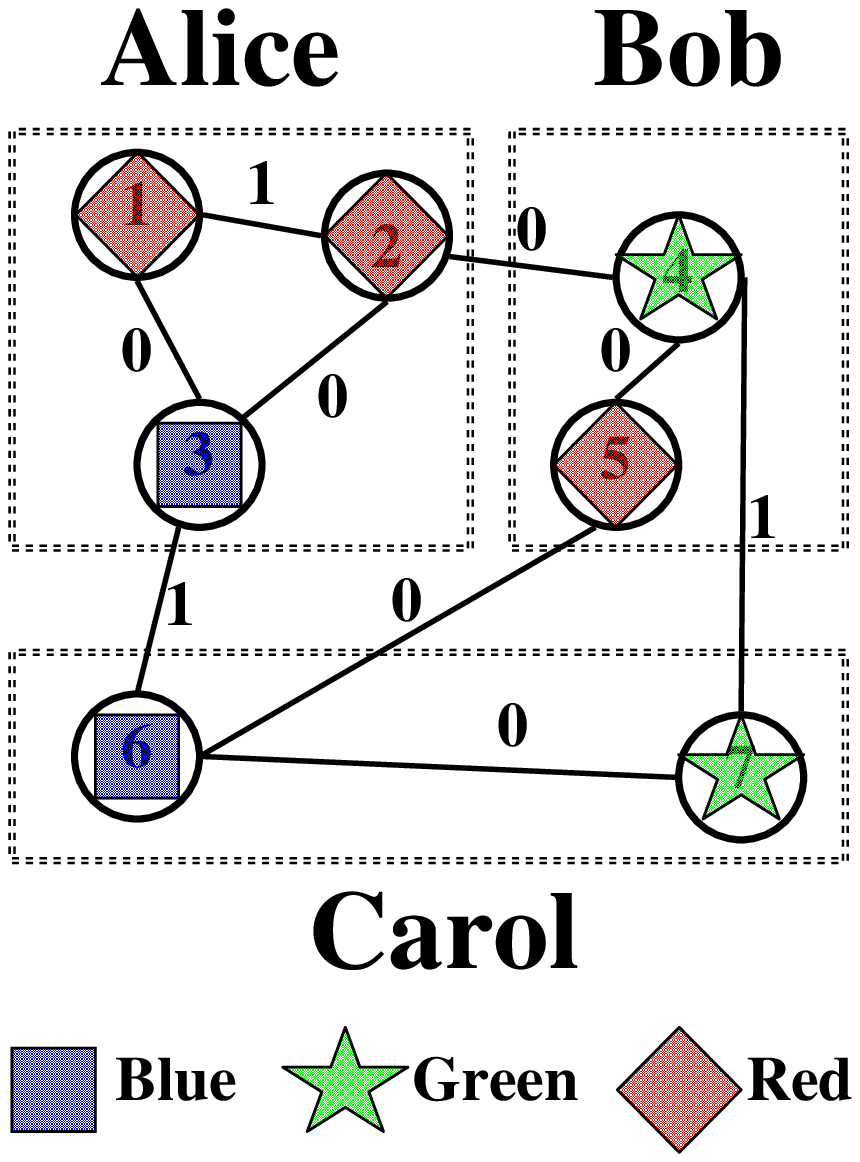} \label{fig:svc} }\hspace{0.4in}
\subfigure[Each Party's Share in
$\mu(x)$=3]{\centering\includegraphics[angle=0,
width=0.52\linewidth]{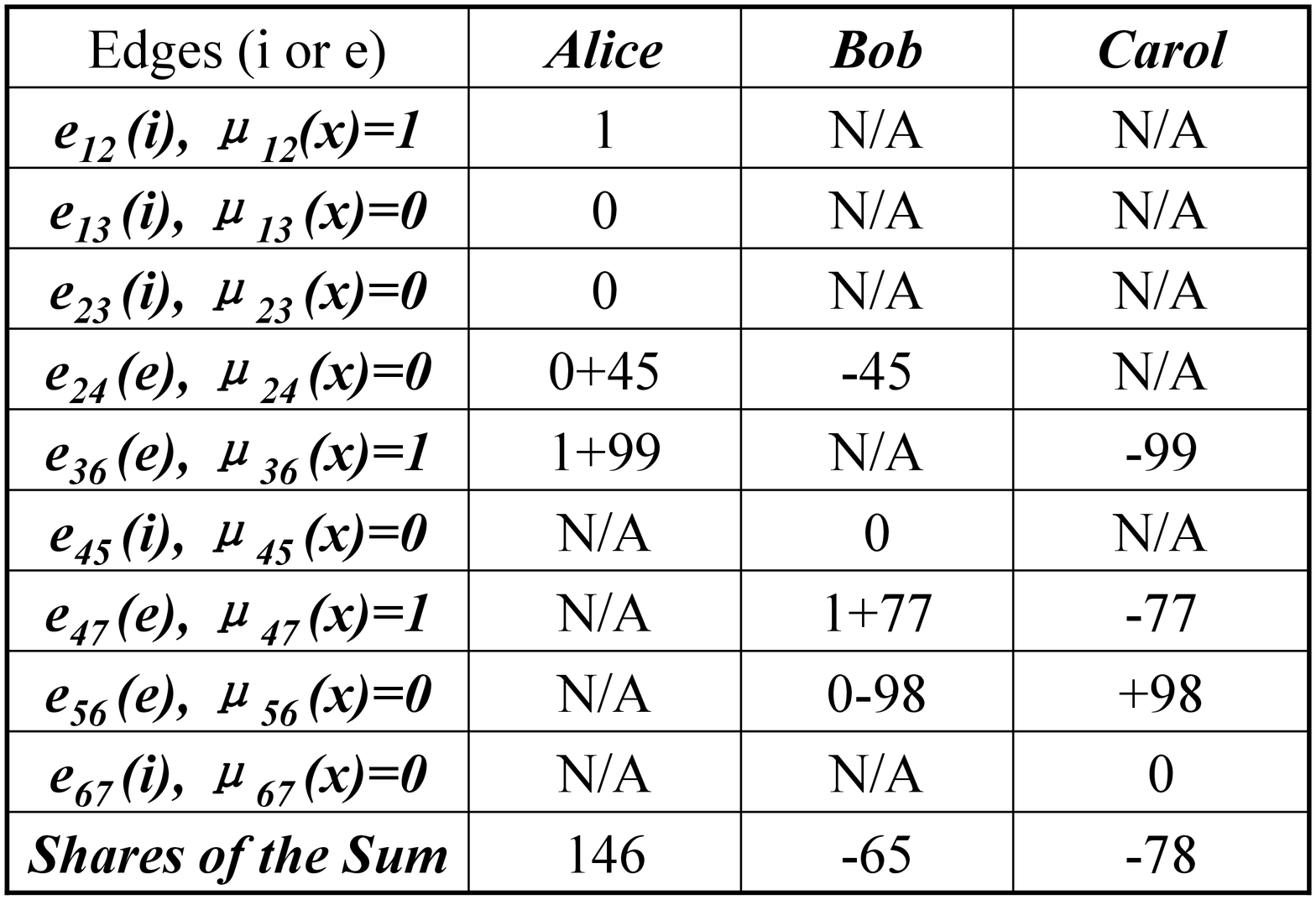} \label{fig:svc2} }
\caption[Example of Distributed Combinatorial Optimization] {Example
for Secure Conflict Computation} \label{fig:scc}
\end{figure}

Fig. \ref{fig:svc} and \ref{fig:svc2} provide a simple example of Secure Conflict Computation. For each external edge, such as $e_{24}, e_{36}$, $e_{47}$ and $e_{56}$, the scalar product computation is
carried out between the two involved parties. As shown in Fig. \ref{fig:svc2}, the returned value in secure scalar product computation is the sum of $\mu_{ij}(x)$ and a random number $r_{ij}$ generated by the other party (Note: the scalar product of every external edge is securely computed only once). To generate each party's output, each party then sums all the returned values with its number of
conflicting internal edges, and finally subtracts all its generated random numbers (the sum of all parties' outputs should be $\mu(x)$). This protocol is provably secure.

\subsubsection{Secure Comparison}

Given two solutions $x$ and $x'$, we can let each party $P_a$ generate its shares $\mu_a(x)$ and $\mu_a(x')$ in $\mu(x)$ and $\mu(x')$ respectively, ensuring that $\mu(x)=\sum_{\forall a}\mu_a(x)$ and $\mu(x')=\sum_{\forall a}\mu_a(x')$. However, $\mu(x)$ and $\mu(x')$ cannot be revealed to any party because they may lead to additional privacy breach in generating neighborhoods.

\begin{algorithm}[!h]
	\begin{algorithmic}[1]
		\renewcommand{\algorithmicrequire}{\textbf{Input:}}
		\renewcommand{\algorithmicensure}{\textbf{Output:}}
		\REQUIRE $G=(V,E)$, partitioned sum of scalar products $\mu(x)$ and
		$\mu(x')$ held by $m$ parties;
		
		\ENSURE $\mu(x')\geq \mu(x)$ or $\mu(x')< \mu(x)$;
		
		\COMMENT{Each party $\forall a\in[1,m], P_a$ has a share of $\mu(x)$
			and $\mu(x')$: $\mu_a(x)$, $\mu_a(x')$}
		
		\FORALL{parties $P_a, a\in[1,m]$}
		
		\STATE Secure Comparing $\mu(x)$ and $\mu(x')$ using its share
		$\mu_a(x)$ and $\mu_a(x')$;
		
		\ENDFOR
		
	\end{algorithmic}
	\caption{$SecCompare(\mu(x'),\mu(x),m)$}\label{algo:sc}
\end{algorithm}

To securely compare $\mu(x')$ and $\mu(x)$ using all parties' partitioned shares (as shown in Algorithm \ref{algo:sc}), we can adopt Yao's secure comparison protocol \cite{Yao86} that compares any two integers securely. Yao's generic method is one of the most efficient methods known. Recently, FairplayMP \cite{fairplaymp} is developed for secure multiparty computation where $m$ untrusted parties with private inputs $x_1,\dots,x_m$ wish to jointly compute a specific functions $f_i(x_1,\dots,x_m)$ and securely obtain each party's output. Specifically, FairplayMP is implemented by writing functions with a high-level language -- Secure Function Definition Language (SFDL) 2.0 and describing the participating parties using a configuration file. The function is then complied into a Boolean circuit, and a distributed evaluation of the circuit is finally performed without revealing anything. In our secure comparison, each party has the shares in $\mu(x)$ and $\mu(x')$ respectively; the function should be $\mu(x')-\mu(x)$ (computed by all $m$ parties); the output $\mu(x')<\mu(x)$ or $\mu(x')\geq \mu(x)$ is generated for designated parties. Finally, we have to compare $\mu(x)$ and $0$ to decide whether to terminate the optimization or not. With the same protocol, we can let all parties securely compare $\mu(x)$ and $1$ with inputs from all parties (the shares of $\mu(x)$ and $1$). The output $\mu(x)<1$ or $\mu(x)\geq 1$ is similarly generated for corresponding parties.

\subsection{Privacy Preserving Tabu Search}
\label{sec:ppls}

\subsubsection{Distributed Tabu List}
In each iteration of Tabucol algorithm, a set of neighborhoods is generated for current solution $x$ in local search. While generating a new neighborhood $x'$, if the new vertex color in $x'$ has already been considered as a taboo in previous iterations, $x'$ should be abandoned except in the case that that $\mu(x')<\mu(x)$ (aspiration criteria). Thus, checking whether the tabu list $T$ includes a move $[x\rightarrow x']$ (we denote the color change for only one vertex as a move) is a key step in privacy-preserving tabu search.

Like Hertz\cite{Hertz}, we define every taboo in $T$ as a visited color for a specific vertex. If all parties
share a global tabu list, the taboos in $T$ may reveal private information (each parties' vertices and unallowed colors). To protect such information, we can let each party $P_a, a\in[0,m-1]$ establish a tabu list $T_a$ to record the latest color of one of its vertices that leads to a critical move ($\mu(x)$ reduced move). We thus have $T=\bigcup_{\forall a} T_a$. Intuitively, since each party manipulates its own tabu list, distributed tabu list does not breach privacy.

\subsubsection{Inference Attack in the Outputs of Local Search}

Since tabu search is a local search based technique, the computation in local search (generating neighborhoods) should be secured. However, even if computing number of conflicts and comparing $\mu(x)$ and $\mu(x')$ are secured by homomorphic encryption and secure comparison, potential inference attack still exists in the outputs of secure computation adapted local search.

In a distributed graph (as shown in Fig. \ref{fig:vertex}), an edge is either internal (belongs to one party) or external (belongs to two parties). Similarly, we can divide the vertices into two complementary groups: 1. \textbf{Inner Vertex}: a vertex that is not on any external edge (i.e. $v_1$ in Fig. \ref{fig:svc}) 2. \textbf{Border Vertex}: a vertex that is on at least one external edge (might be on both of external edges and internal edges, i.e. $v_2$ in Fig. \ref{fig:svc}).

We first assume that no secure comparison is applied and $v_i$ is the color changed vertex in generating $x$'s neighborhood $x'$ ($\mu(x)$ is revealed in each iteration). If $v_i$ is an inner vertex, since $v_i$ is disconnected with other parties' vertices, nothing is revealed to each other. However, if $v_i$ is a border vertex, we first look at a simple case: $v_i$ is an endpoint of only one external edge (i.e. $v_5$ in Fig. \ref{fig:vertex}). Assume that $P_a$ and $P_b$ hold $v_i$ (color $x_i$) and $v_j$ (color $x_j$) respectively where $e_{ij}$ is the external edge. While changing $v_i$'s color $x_i$ to $x_i'$, since the number of $P_a$'s conflicting internal edges in $x$ and $x'$ are known to $P_a$ and the total number of conflicts on the remaining edges (including all the external edges except $e_{ij}$ and other parties' internal edges) are identical in $x$ and $x'$, the difference between $\mu_{ij}(x)$ ($e_{ij}$ has conflict or not in solution $x$) and $\mu_{ij}(x')$ ($e_{ij}$ has conflict or not in solution $x'$) can be inferred by $P_a$. We thus discuss the potential privacy breach in three cases where $\delta=\mu_{ij}(x')-\mu_{ij}(x)$:

\begin{enumerate}
\item $\delta=1$. $P_a$ can infer $x_j=x_i'$. The color of $v_j$ is learnt by $P_a$. 

\item $\delta=0$. $P_a$ can infer $x_j\ne x_i$ and
$x_j\ne x_i'$. It is secure for large $k$. For small $k$, $P_a$ can
guess $x_j$ as any color out of remaining k-2 colors except $x_i$
and $x_i'$.

\item $\delta=-1$. $P_a$ can infer $x_j=x_i$. The color of $v_j$ is learnt by
$P_a$. 

\end{enumerate}

In a more general case, if $v_i$ is adjacent to $n$ border vertices from other parties and no secure comparison is implemented in tabu search, moving $v_i$ to generate neighborhoods is under the risk of inference attack:

\begin{lemma}
Given a coloring solution $x$, we generate $x$'s neighborhood $x'$ by changing the color of a border vertex $v_i$ of party $P_a$ which is adjacent to $n$ vertices $\{\forall j\in[1,n], v_j\}$ from other parties. If no secure comparison is implemented in tabu search ($P_a$ knows $\delta=\sum_{\forall j\in[1,n]}[\mu_{ij}(x')-\mu_{ij}(x)]$), precisely one of the following holds:

\begin{enumerate}
	\item if $\delta\geq 0$, $P_a$ can guess the color of $v_j$ as $x_j=x_i'$ with $Prob[x_j=x_i']=
	\frac{\delta+\lfloor\frac{n+\delta}{2}\rfloor}{2n}$ and guess $x_j=x_i$ with $Prob[x_j=x_i]=
	\frac{\lfloor\frac{n-\delta}{2}\rfloor}{2n}$ where $\delta\in[0,n]$;
	
	\item if $\delta<0$, $P_a$ can guess the color of $v_j$ as $x_j=x_i$ with $Prob[x_j=x_i]=
	\frac{-\delta+\lfloor\frac{n-\delta}{2}\rfloor}{2n}$ and guess $x_j=x_i'$ with $Prob[x_j=x_i']=
	\frac{\lfloor\frac{n+\delta}{2}\rfloor}{2n}$ where $\delta\in[-n,-1]$.
	
\end{enumerate}

\label{lemma:pri}
\end{lemma}

We prove Lemma \ref{lemma:pri} in Appendix \ref{sec:lem1} (we do not discuss the probability of guessing $x_j$ as any color out of remaining k-2 colors except $x_i$ and $x_i'$ since the probability is fairly small in general). Hence, if $\mu(x)$ is revealed in each iteration (no secure comparison protocol is applied), the probability of guessing any other parties' border vertex color is high. For instance, assuming that $x_i=\textit{Blue}$, $x_i'=\textit{Red}$, $\delta=3$ and $n=5$, $P_a$ can guess that $x_j$ is \textit{Red} with a probability $0.7$ and is \textit{Blue} with a probability $0.1$. If $P_a$ iteratively moves $x_i$ and learns the probabilities in such process, the color of $v_j$ should
be almost determined. Furthermore, if $\delta$ is close to either $n$ or $-n$, the probability of inferring privacy is extremely high, especially in case that $\delta=n$ or $-n$ (the worst privacy breach case for any border vertex on $n$ external edges). Alternatively, if we utilize secure comparison protocol to securely compare $\mu(x')$ and $\mu(x)$ in each iteration, such probabilities can be
comparatively reduced (especially in case that $\delta$ is close to $n$ or $-n$) because $\delta$ is anonymized as any integer in [$-n,n$].

\footnotetext[1]{$\delta_a$ denotes the difference between the number of $P_a$'s conflicting internal edges in solution $x$ and $x'$, which is known to $P_a$ in every iteration}

\begin{lemma}
In the same circumstance of Lemma 1, if we securely compare $\mu(x)$ and $\mu(x')$ (does not reveal $\mu(x)$ and $\mu(x')$), the risk of inference attack can be reduced in general but still exists in two
worst cases with known $\delta_a$ \footnotemark[1] as below:

\begin{enumerate}
\item If $\delta_a\footnotemark[1]=n-1$ and $\mu(x')<\mu(x)$ is the output of the secure comparison, $P_a$ learns $\forall j\in[1,n],x_j=x_i$;

\item If $\delta_a\footnotemark[1]=-n$ and $\mu(x')\geq\mu(x)$ is the output of the secure comparison, $P_a$ learns $\forall j\in[1,n], x_j=x_i'$.

\end{enumerate}

 \label{lemma:sc_pri}
\end{lemma}

\begin{figure*}[h]
\centering \subfigure[Worst Case
1]{\centering\includegraphics[angle=0,
width=0.5\linewidth]{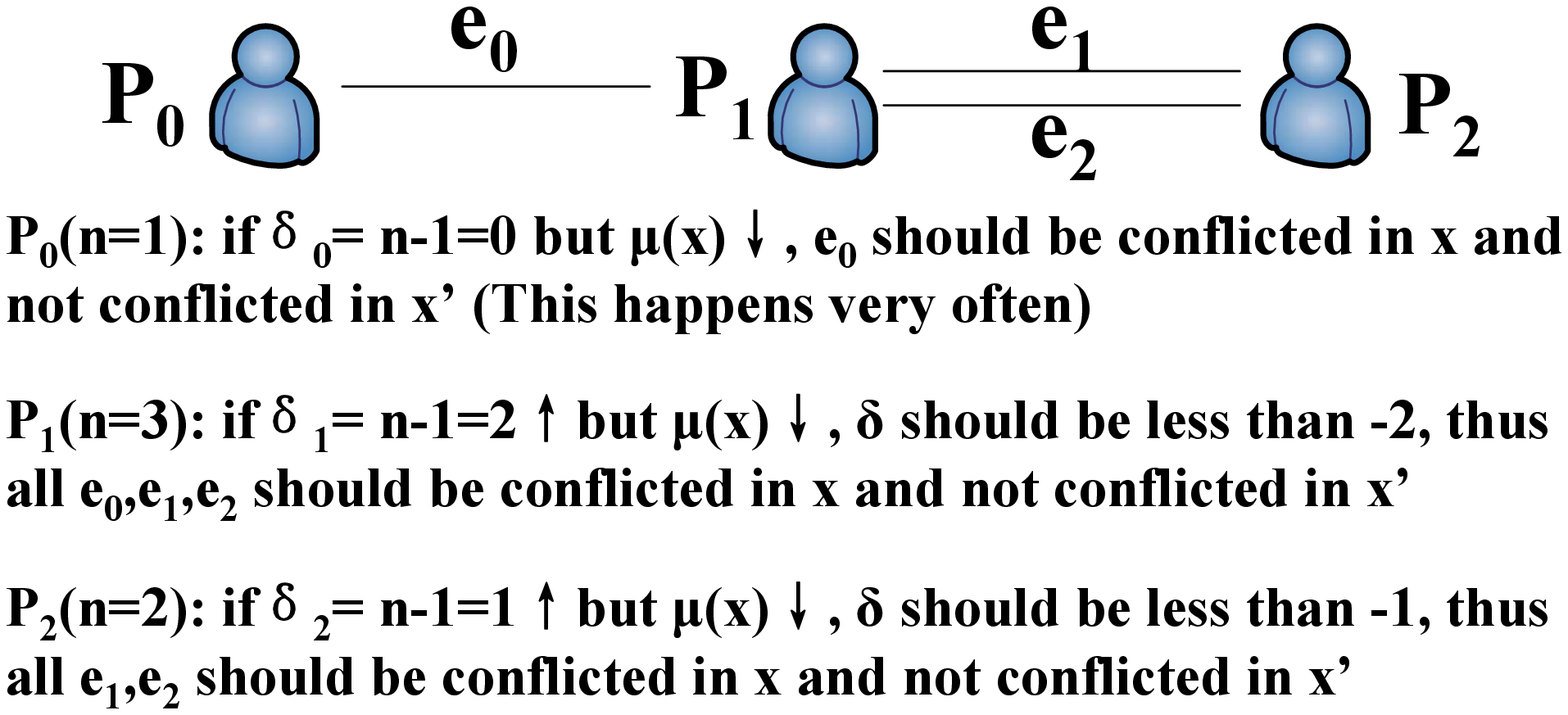} \label{fig:wcase1}
}\hspace{-0.1in}\subfigure[Worst Case 2]{\centering\includegraphics[angle=0,
width=0.5\linewidth]{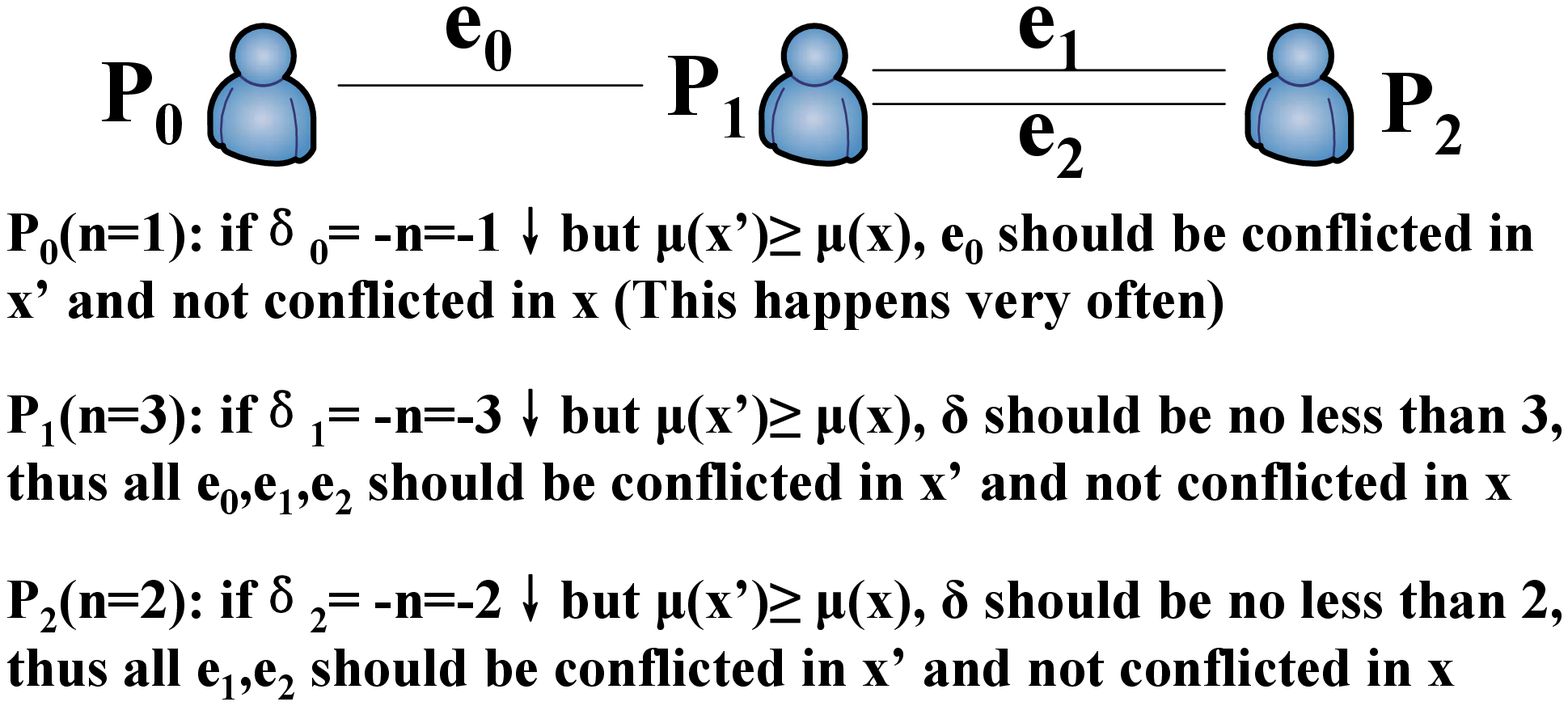} \label{fig:wcase2} } \caption{An
Example for Inference Attack on Secure Comparison [$P_1$(n=3) has a
border vertex that is adjacent to one vertex of $P_0$(n=1) and two
vertices of $P_2$(n=2)]} \label{fig:wcase}
\end{figure*}

As described in Lemma \ref{lemma:sc_pri} (proved in Appendix \ref{sec:lem2}), the worst privacy breach cases cannot be resolved by secure comparison protocol. We show an example for two worst cases in Fig. \ref{fig:wcase}. While changing the color of $P_0$, $P_1$, or $P_2$'s one border vertex to generate neighborhoods, possible privacy breach occurs in these worst cases. For example, $P_0$ moves its border vertex: if the number of $P_0$'s conflicting internal edges has not been changed in the move $x\rightarrow x'$ and $\mu(x')<\mu(x)$ is the comparison result (both are known to $P_0$), the only external edge $e_0$ should be conflicted in $x$ and not conflicted in $x'$. Thus $P_0$
can infer that the other endpoint's color ($P_1$'s one border vertex) is the same as its border vertex's color in $x$ (before move). This happens very often.

Therefore, we will try to eliminate such inference attack and reduce the guessing probability even more in the following.

\subsubsection{Synchronous Move in Local Search}

Since the inference attack results from the output that should be revealed in each iteration (the result of comparing $\mu(x)$ and $\mu(x')$), rather than the data within the secured computation function, it is extremely difficult to resolve this issue with a secure protocol in which such comparison results are not revealed for selecting better solution step by step (and still retaining good efficiency). Alternatively, we can eliminate such inference attack and reduce the probabilities on guessing conflicted external edges and the colors of other parties' border vertices by slightly modifying local search.

Specifically, if any border vertex $v_i$ is chosen for move while generating neighborhoods for one party, we can let another party to locally change the color of one of its vertex (either inner or border) or opt to \emph{keep invariant}. We denote this mechanism as \textbf{synchronous move}. The inferences can be eliminated and the guessing possibilities can be greatly reduced with synchronous move (we prove it in Appendix \ref{sec:lem3}).

\begin{lemma}
	Synchronous move resolves the inference attack stated in Lemma \ref{lemma:pri} and \ref{lemma:sc_pri}.
\end{lemma}

Apparently, some generated neighborhoods may include two moves. In addition, since the number of moves is at most 2, it is still effective to maintain the performance of local search. Hence, the performance is supposed to be close to Tabucol algorithm \cite{Hertz} (we show that in the experiments).

To improve the search performance, the vertex (either inner or border vertex) whose color is changed is \textit{preferentially chosen from the vertices that are the endpoints of edges with conflicts}. Since only the conflicting internal edges are known to corresponding parties, we can iteratively let each party change the color of the endpoints of its conflicting internal edges. When moving an inner vertex, the party can independently compare two unknown numbers $\mu(x')$ and $\mu(x)$ (by comparing its shares in $\mu(x')$ and $\mu(x)$) and check its tabu list along with specific operations on its vertices colors in $x'$ (update $x$ and tabu list with $x'$, or abandon $x'$). On the other hand, moving a border vertex requires secure computation of all scalar products and secure comparison of $\mu(x')$ and $\mu(x)$.  Hence, while changing a border vertex color by one party, our privacy preserving tabu search protocol calls synchronous move, secure conflict computation and secure comparison. The result of secure comparison is only revealed to the two involved parties. These parties use this information to check and update their tabu lists along with specific operations on their
vertices' colors in $x'$.

\subsubsection{Privacy Preserving Tabu Search}

\begin{figure*}[tb]
\centering
\includegraphics[angle=270, width=1\linewidth]{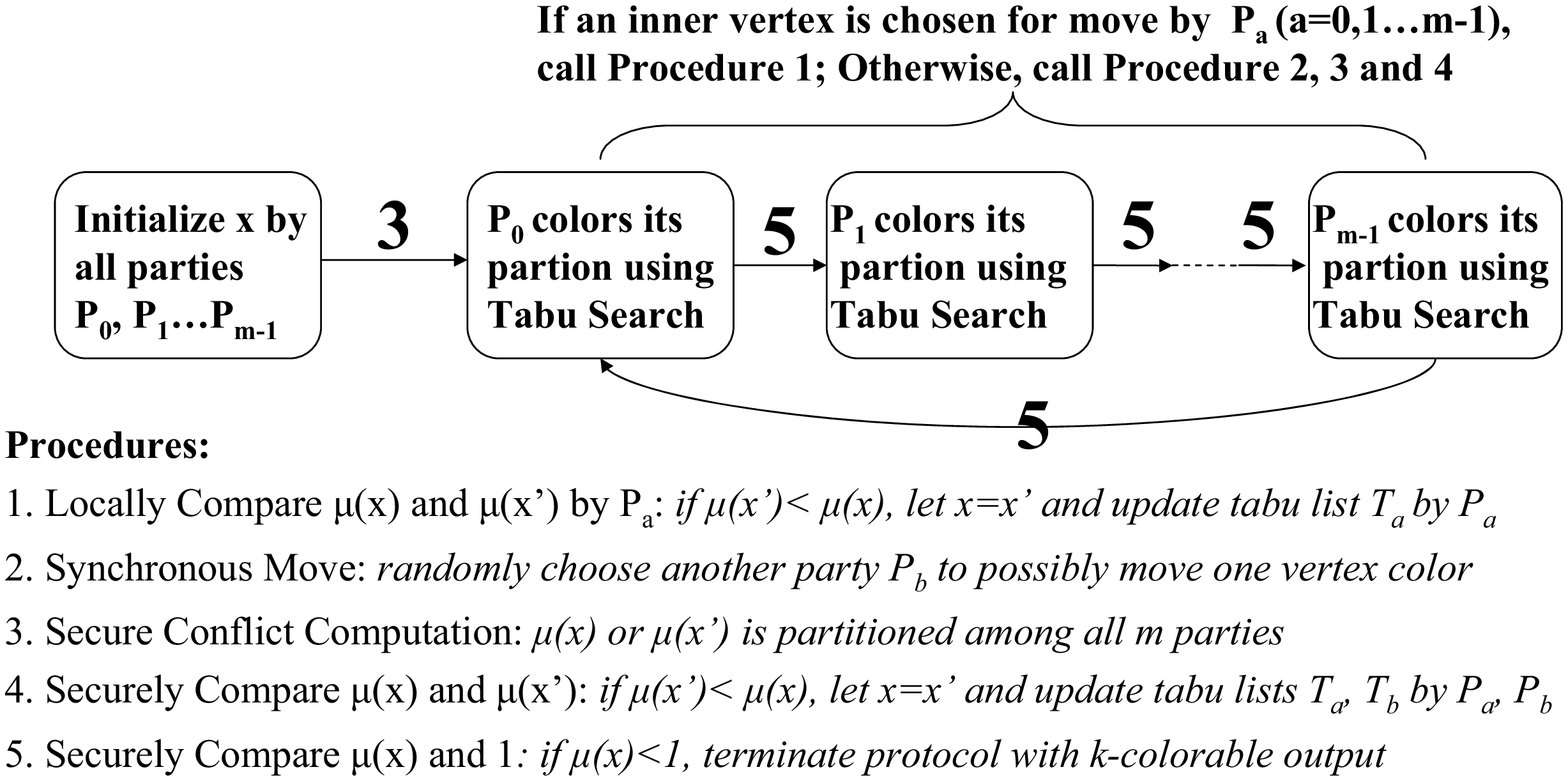}
\label{fig:sparse} \caption[Example of Distributed Combinatorial
Optimization] {Privacy Preserving Tabu Search Protocol}
\label{fig:ppts}
\end{figure*}

In our secure tabu search protocol, each party iteratively executes Tabu search on its local graph to eliminate all the conflicting internal edges. (Each party's local graph  is a subgraph of $G$. Hence, if the local graph is not $k$-colorable, the protocol can be terminated with output that $G$ is not $k$-colorable). During local search, each party moves inner or border vertices (with synchronous move).  If the party executing tabu search in the given iteration does not have any conflicting internal edge (local k-coloring is attained), this party may move  one of its border vertices that triggers  synchronous move, or skip its local search and pass control to the next party. In addition, in each move, tabu list is locally checked and aspiration criteria is applied to every solution (locally
for inner vertex move or globally for synchronous move in two parties' tabu lists). Finally, in each iteration after a party finishes k-coloring of its local graph, secure comparison protocol is called to check if an optimal solution is obtained (i.e., $\mu(x)=0$). Since the output of secure comparison is divided into $<$ and $\geq$, the algorithm compares $\mu(x)$ with $1$. If $\mu(x)=0$, the secure tabu search is terminated with $k$-colorable output. If $\mu(x)\geq 1$ after a given number of iterations, the tabu search is terminated with an output that $G$ is not $k$-colorable. The detailed steps are presented in Fig. \ref{fig:ppts} and Algorithm \ref{algo:ppts}.

\begin{algorithm}[!h]
\begin{algorithmic}[1]

\renewcommand{\algorithmicrequire}{\textbf{Input:}}
\renewcommand{\algorithmicensure}{\textbf{Output:}}

\REQUIRE $m$ parties $P_0,\dots,P_{m-1}$, Graph $G=(V,E)$;

\STATE initialize a solution $x$: randomly coloring all vertices by
all parties;

\STATE call Secure Conflict Computation: each party hold a share
of $\mu(x)$: $\forall a\in[0,m-1], \mu_a(x);$

\FOR{party $P_a$, $a\in[0,m-1]$}

\WHILE{the number of $P_a$'s conflicting internal edges is greater
than $0$}

\STATE $P_a$ generates a neighborhood $x'$ by changing the color of
vertex $v_i$ that is the endpoint of one of its conflicting internal
edges: $x_i\rightarrow x_i'$;

\IF{$v_i$ is an inner vertex}

\STATE $P_a$ locally compare its shares in $\mu_a(x)$ and
$\mu_a(x')$ and check its tabu list $T_a$;

\IF{$\mu_a(x')<\mu_a(x)$ \footnotemark[2]}

\STATE let $x=x'$, update the tabu list $T_a$ and continue;

\ENDIF

\ELSE

\STATE pick a random party $P_b$ ($b\ne a$) to possibly change the
color of vertex $v_j$: $x_j\rightarrow x_j'$ ($x'$ is generated by
synchronous move);

\STATE call Secure Conflict Computation and Secure Comparison for
$\mu(x')$ and $\mu(x)$;

\IF{$\mu(x')<\mu(x)$ \footnotemark[3]}

\STATE let $x=x'$, update the tabu list $T_a$, $T_b$ by $P_a$, $P_b$
respectively and continue;

\ENDIF

\ENDIF

\STATE if $x$'s $rep$ neighborhoods are generated ($rep$ is given),
choose a best neighborhood $x'$ where $(v_i,x_i')\notin T_a$, let
$x=x'$ and update $T_a$ ($v_i$ is an inner vertex of $P_a$);

\STATE if $P_a$'s partition in $G$ cannot be $k$-colored in maximum
number of iterations, terminates protocol with output: $G$ is not
k-colorable;

\ENDWHILE

\STATE if $P_a$'s partition is $k$-colored, $P_a$ can move its
border vertices to generate new $k$-coloring solutions (call synchronous move, secure conflict computation and secure comparison);

\STATE call Secure Comparison to compare $\mu(x)$ and 1 among all
$m$ parties;

\STATE if $\mu(x)<1$, terminates protocol with output: $G$ is
k-colorable by solution $x$;

\STATE if $\mu(x)\geq1$ after a given number of iterations,
terminates protocol with output: $G$ is not k-colorable;

\STATE $a=(a+1)\mod m$;

\ENDFOR

\end{algorithmic}
\caption{Privacy Preserving Tabu Search}\label{algo:ppts}
\end{algorithm}

\footnotetext[2]{Aspiration Criteria: if $\mu(x')<\mu(x), (v_i, x_i')\in T_a$ is allowed; Otherwise, abandon $x'$.}

\footnotetext[3]{Aspiration Criteria: if $\mu(x')<\mu(x), (v_i,x_i')\in T_a$ and $(v_j,x_j')\in T_b$ are allowed; Otherwise, abandon $x'$.}

\section{Analysis}
\label{sec:security}
In this section, we give analysis for security and complexity.

\subsection{Security Analysis}

The exchanged messages of PPTS protocol occurs primarily in two sub-protocols: secure scalar product computation and secure comparison. Thus, the security of our algorithm depends on the
security of these two algorithms. Given these secure protocols, the security proof comes down to demonstrating that the output of those protocols can be simulated knowing one's own input and the final results. The composition theorem \cite{Goldreichenc} then enables completing the proof of security.

\begin{theorem}
PPTS privately solves the $k$-coloring problem where each party $P_a$ learns:

\begin{enumerate}
\item Its share in the secure conflict computation

\item The secure comparison result in synchronous moves involving $P_a$
\end{enumerate}

 \label{thm:sec}
\end{theorem}

\begin{proof}
Since PPTS protocol is symmetric, proving that the view of one party can be simulated with its own input and output suffices to prove it for all parties. We now show how to simulate the messages received by an arbitrary party $P_a, a\in[0,m-1]$. While $P_a$ is chosen to color its partition (this is iteratively done in PPTS), two cases occur.

First, if $P_a$ changes the color of its inner vertices (Step 6-10), $P_a$ learns only $\mu(x')< \mu(x)$ or $\mu(x')\geq \mu(x)$ by locally comparing its shares in $\mu(x')$ and $\mu(x)$. Other parties learn nothing (the number of iterations to color $P_a$'s partition is also unknown to other parties). Hence, these can be simulated by simply executing those steps.

Second, if $P_a$ moves one of its border vertices, a synchronous move is triggered. $P_a$ now receives the share of secure conflict computation: $s_{ij}$ (w.o.l.g. we let $P_a$ be the receiver of $s_{ij}$ rather than the creator of $r_{ij}$), and learns the secure comparison result of two solutions ($x$ and $x'$): $\mu(x')<\mu(x)$ or $\mu(x')\geq \mu(x)$. We now examine $P_a$'s view in Step 11-17 (and possibly include Step 21) and build a polynomial simulator that runs for the known synchronous moves.

\begin{figure}[h]
\centering\includegraphics[angle=0, width=1\linewidth]{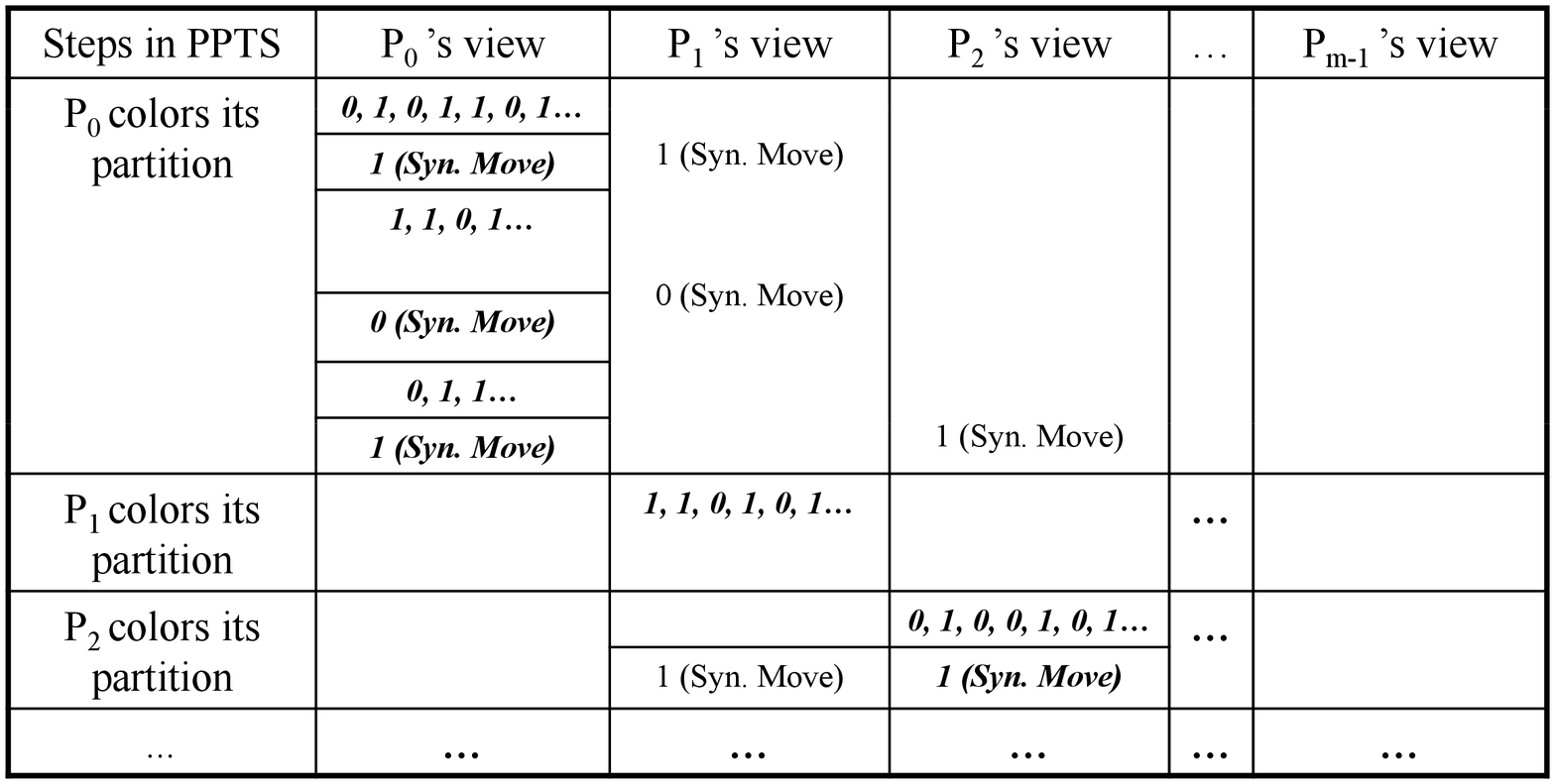}
\caption{$P_0,\dots,P_{m-1}$'s View (the sign of
$\mu(x')-\mu(x)$: 1 is $\mu(x')<\mu(x)$; 0 is $\mu(x')\geq\mu(x)$)}
\label{fig:move}
\end{figure}

\textit{$P_a$'s view}: in order to simulate $P_a$'s view, the polynomial simulator must satisfy the following condition: given $P_a$'s partition (input) and the final solution (output), simulate the views of secure conflict computation and secure comparison in every synchronous move (the shares in secure scalar product computation and a sequence of results of secure comparison  as shown in Figure \ref{fig:move}). The simulator is built as below:

\begin{itemize}
\item The result of secure scalar product computation $s_{ij}=x_i\cdot x_j+r_{ij}$ (w.o.l.g. we let $P_a$ be the receiver of $s_{ij}$), can be simulated by generating a random from the uniform probability
distribution over $\mathcal{F}$ (Note: we assume that $s_{ij}$ are scaled to fixed precision over a closed field, enabling such a selection). Thus, $Prob[T=t]=Prob[s_{ij}=t]=Prob[r_{ij}=t-s_{ij}]=\frac{1}{\mathcal{F}}$.

\item To simulate the sequence of secure comparison results for $P_a$, we can utilize the inverse steps of tabu search on $P_a$'s border vertices (we denote this process as inverse optimization). A key result from inverse optimization is that if the original problem is solvable in polynomial time, the reverse problem is also solvable in polynomial time \cite{AhujaInverseOpt}. Specifically, letting the cost function be maximizing $\mu(x)$ ($x'$ remains the neighborhood of $x$), we start from the final colorable solution (only in $P_a$'s known part of the graph)\footnotemark[4], generate neighborhoods by changing the color of its border vertices (color is randomly chosen), and compare $\mu(x')$ and $\mu(x)$. If $\mu(x')>\mu(x)$, the simulator outputs an ``1'', otherwise it outputs a ``0'' (Since the comparison signs in moving these vertices in the simulator are deterministic and known to $P_a$, the comparison signs of synchronous move in original PPTS protocol can be simulated by this machine). Although the colors and the vertices might be different from the original
one, the sign of comparison is the same. 
\end{itemize}

\footnotetext[4]{Since $P_a$ knows the existence of some adjacent border vertices from other parties, we can let $P_a$ color them at the beginning. Thus, moving $P_a$'s any vertex is deterministic in
the simulator.}

\textit{$P_b$'s view}: First, $P_b$ knows the sign of $\mu(x')-\mu(x)$ in all synchronous move in which it is involved. Similarly, we can use the simulator to generate the signs of comparison as above. Secondly, in all secure scalar product computation, the random number $r_{ij}$ can be generated again using the same uniform distribution in secure conflict computation. Since the outputs of other parties are similar to $P_b$, we can use the same simulator.

In the remaining Steps of PPTS, each party redefines new current solution and manipulates the tabu list in local computation, thus we can simulate them by executing those steps. Therefore, PPTS is secure against semi-honest behavior, assuming that all the parties follow the PPTS protocol.

Thus, Theorem \ref{thm:sec} has been proven. Note that, effectively, a party can learn at most the color of the border vertex of another party with which it shares an edge. This is useless in the intermediate stages (since the color is not fixed). In the final solution, since typically the number of shared border vertices is significantly smaller than the total number of vertices, the leaked information is not very significant.
\end{proof}

There are several remaining issues. First, since our simulator is also a meta-heuristic, the worst case running time is not polynomial. This is a problem-dependent issue: if the problem is always solved by tabu search in polynomial time, then the simulator is guaranteed to be polynomial, and the protocol is secure. Second, some additional minor information is probably leaked. For example, in the output of PPTS, given $k$-colorable solution, the colors of $P_a$'s adjacent vertices (belong to other parties) can be known (to $P_a$) with a different color from corresponding adjacent vertex. This is unavoidable in the DisGC problem (even if we build a trusted third party, it is learned by other parties with the output). Moreover, the number of iteratively calling different parties to color its partition and the fact that $\mu(x)>0$ until the last step are also revealed. These pieces of information are trivial and do not significantly impact the security of the PPTS protocol.

\subsection{Communication and Computational Cost Analysis}

Given a distributed graph with $n$ vertices and $n_e$ external edges, in PPTS protocol, we assume that $m$ parties are repeatedly called to locally color their partition for $c$ times and $\ell$ synchronous moves are executed during coloring of the partition by a party. Based on this, we can discuss the computation and computational cost of the PPTS protocol.

\subsubsection{Communication Cost}
%In PPTS protocol,
Note that only secure conflict computation and secure comparison require communication of messages among parties. Since conflict computation for each external edge requires two communication messages, the total number of messages generated for secure conflict computation is $2c*(n_e*\ell+1)$, where $O(n_e)=O(n^2)$. Thus, the communication complexity of secure scalar product computation in PPTS is $O(cn^2\ell)$. Moreover, every secure comparison requires $O(m)$ parties to compute the output. Since the number of communication messages in every secure comparison equals to the number of computing parties \cite{Grama03} and PPTS includes $c*(\ell+1)$ secure comparison, the communication complexity of secure comparison in PPTS is $O(cm\ell)$.

\subsubsection{Computation Cost}

Since each party triggers $\ell$ synchronous moves while coloring its partition, $c*\ell+1$ secure conflict computations and $c*(\ell+1)$ secure comparisons are required in PPTS.  If the runtime of a single secure scalar product computation and a single secure comparison is denoted by $t_s$
and $t_c$ respectively, the additional cost can be estimated as $(c*\ell*n_e+1)*t_s+c*(\ell+1)*t_c$. Thus, the computational complexity of PPTS is $O(cn^2\ell)$ for secure scalar product and
$O(c\ell)$ for secure comparison.

\section{Experiments}
\label{sec:exp}

In this section, we present the experimental validation of our algorithm for securely solving distributed graph coloring problems. 

\subsection{Experiments Setup}
For the experiments, we randomly generate graphs with the number of vertices $N$ and the density of the graph $\rho$ (probability of the existence of an edge between any two given vertices). Each generated graph is partitioned among $10$ parties ($P_0,P_1,\dots,P_9$), we let each party hold N/10 vertices. In each test, we generate $10$ graphs with the same $N$ and $\rho$, and report the average of the results.

We compare our privacy preserving tabu search (PPTS) algorithm with Tabucol \cite{Hertz} using the same group of graphs. Both algorithms are terminated if no $k$-colorable solution can be found in $10^5$ iterations. We let each party hold a tabu list with length $N/10$ in PPTS while the overall tabu list in Tabucol is established with length $N/10$. In synchronous move of PPTS algorithm, we assume that $P_b$ (the second party) has 50\% chance of skipping a move. Futhermore, in PPTS algorithm, Paillier cryptosystem \cite{Paillier99} with 512-bit and 1024-bit key is used for homomorphic encryption (securely computing number of conflicts). For FairplayMP \cite{fairplaymp}, sufficiently large modulus length for 10 parties (i.e. 128 bits) is used for secure comparison. All the experiments are performed on an HP machine with Intel Core 2 Duo CPU 3GHz and 6GB RAM.

Our goal is to validate two hypotheses: (1). the optimal results (chromatic number) of solving DisGC problem using PPTS are close to the results of Tabucol (this indicates that our approach is as accurate as any general meta-heuristic for graph coloring, and (2). the computation time of PPTS is reasonable.  Thus, for the first goal, we compute the chromatic number of distributed graphs with varying number of vertices \#V=\{100, 200, 300, 500, 1000\} and for two different densities $\rho$=10\% and 30\% using PPTS and Tabucol (for each pair of \#V and $\rho$, we generate 10 distributed graphs; 512-bit key is used). For the second goal, first, we compute the runtime of PPTS with 512-bit key and with 1024-bit key and compare it with Tabucol for varying \#V=\{100, 200, 300, 500, 1000\} and fixed $\rho$=20\% (Fig. \ref{fig:v}); second, we perform the comparison between the execution time of PPTS and Tabucol for varying $\rho=\{2\%,5\%,10\%,20\%,30\%\}$ and fixed \#V=500 (Fig. \ref{fig:d}). Table \ref{table:set} shows the specified parameters for the experiments.

\begin{table}[h!]
	\centering \caption{Experimental Parameters}
	\begin{tabular}{|l|c|c|}
		\hline
		Parameter & Selected Values & If Fixing\\
		\hline
		$\#$ of vertices &  100, 200, 300, 500, 1000 & 500 \\
		\hline
		Dens. of Edges & 2\%, 5\%, 10\%, 20\%, 30\% & 20\% \\

		\hline
		Key Length & 512-bit, 1024-bit & 512-bit\\
		\hline

	\end{tabular}
	\label{table:set}
\end{table}

\subsection{Experimental Results}

\begin{figure*}[!t]
	\centering \subfigure[\#V=100, $\rho$=10\%,
	k=10]{\centering\includegraphics[angle=0,
		width=0.32\linewidth]{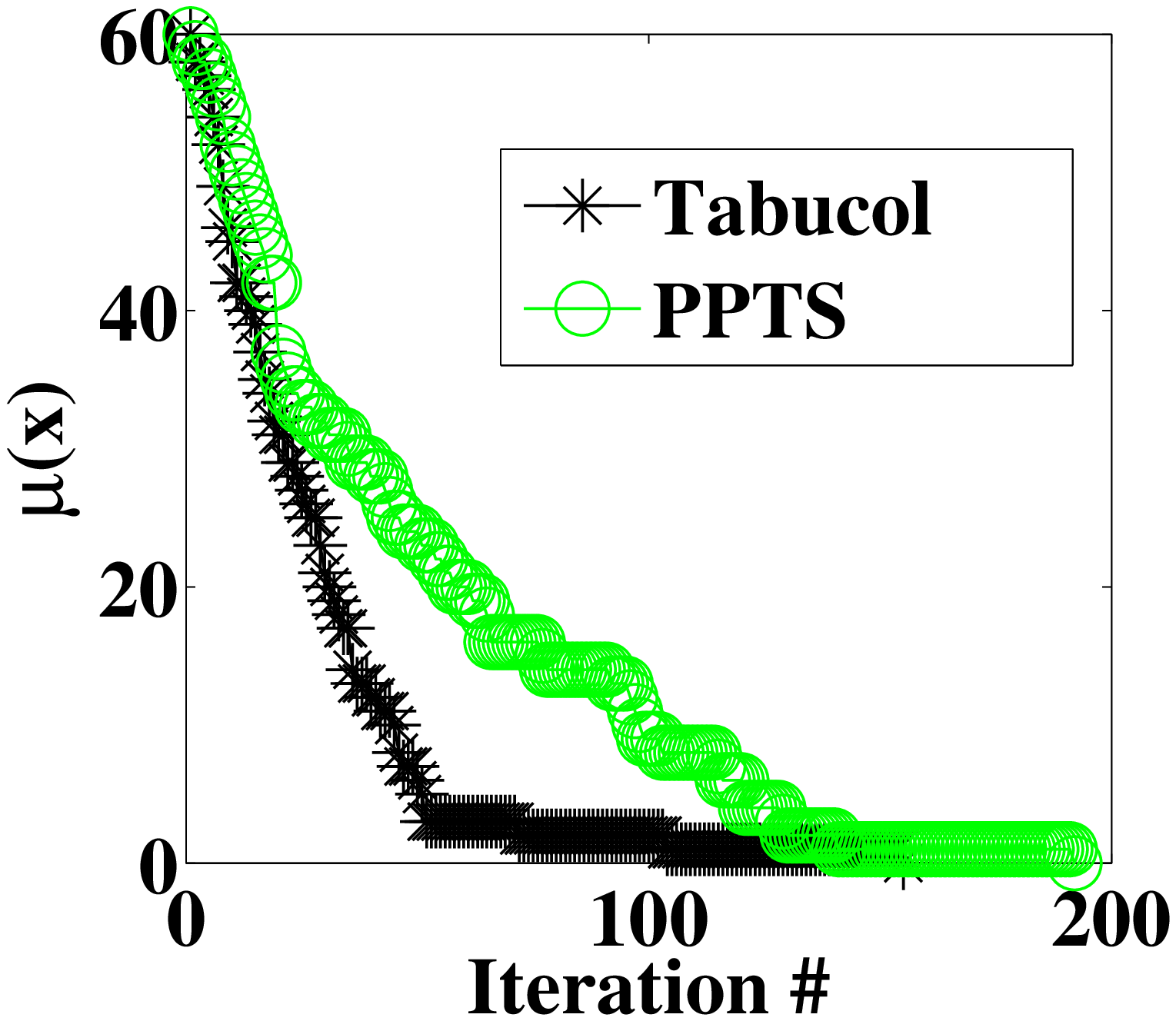} \label{fig:sparseiter}
	}\subfigure[\#V=300, $\rho$=20\%,
	k=30]{\centering\includegraphics[angle=0,
		width=0.32\linewidth]{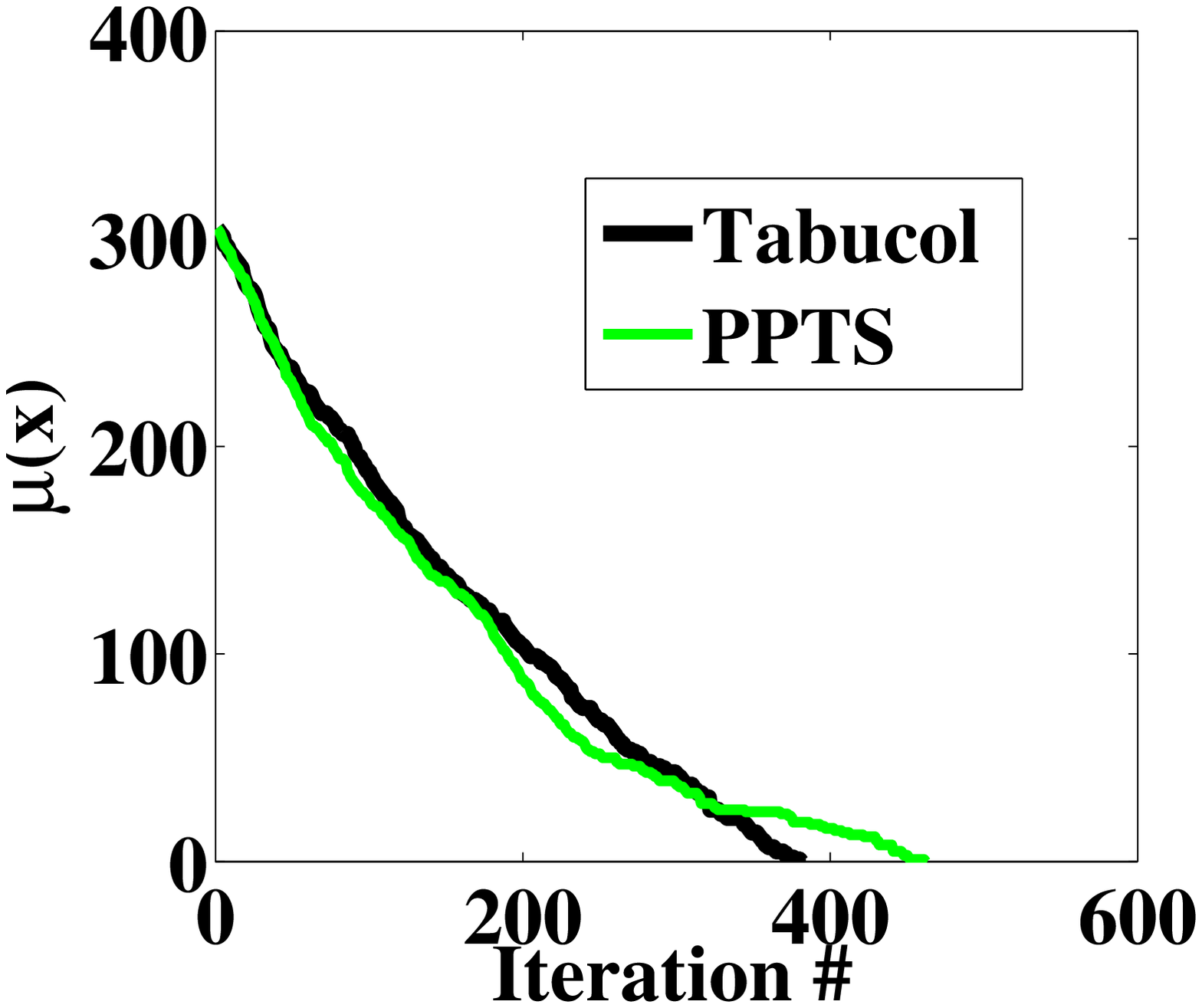} \label{fig:mediter} }
	\subfigure[\#V=500, $\rho$=30\%,
	k=40]{\centering\includegraphics[angle=0,
		width=0.32\linewidth]{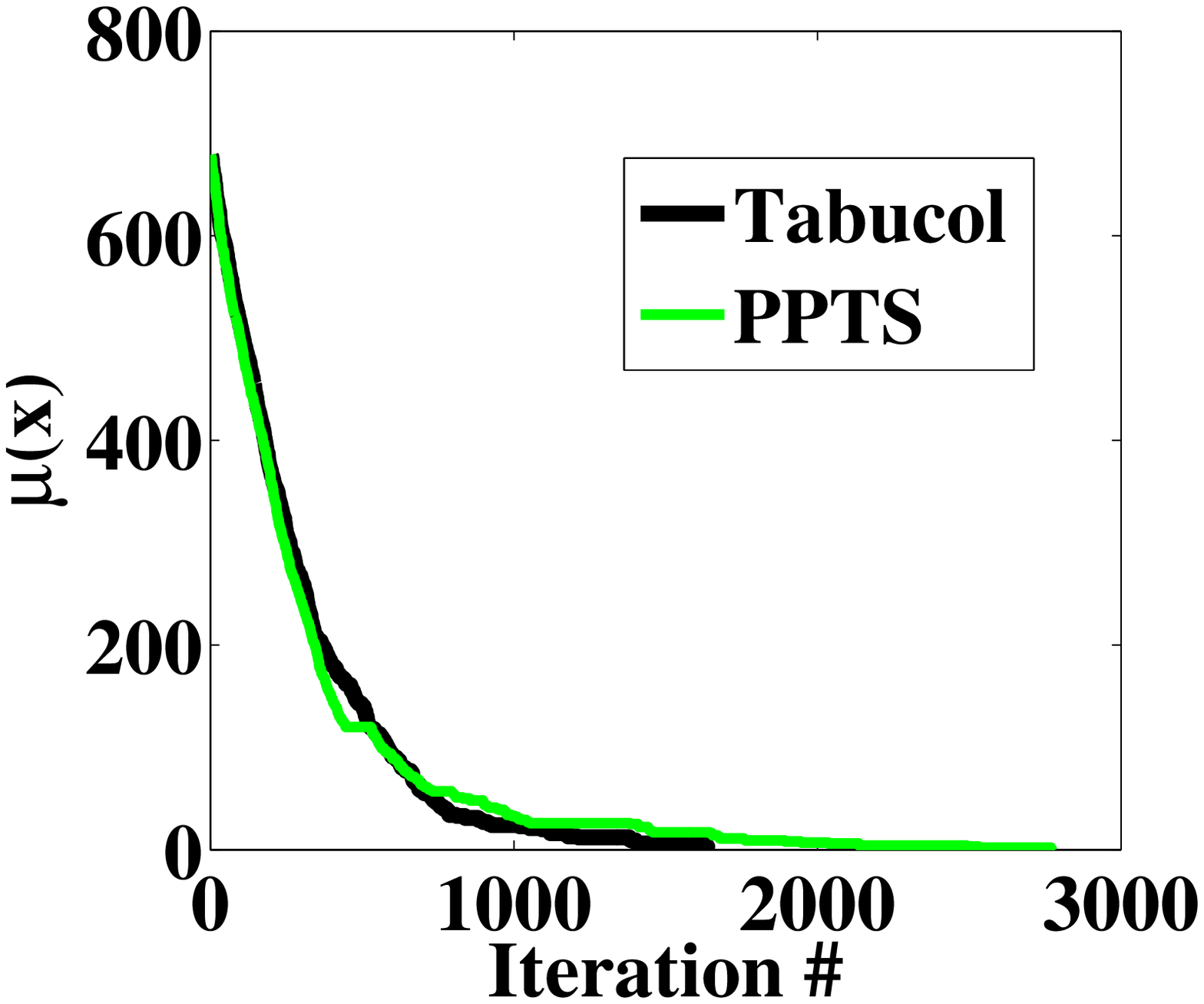} \label{fig:denseiter}
	}\vspace{-0.1in} \caption[Example of Distributed Combinatorial
	Optimization] {Total Number of Conflicts ($\mu(x)$) in each
		Iteration}\label{fig:iter}
\end{figure*}

\begin{figure*}[!t]
	\centering \subfigure[Chromatic \# (min
	k)]{\centering\includegraphics[angle=0, width=0.32\linewidth]{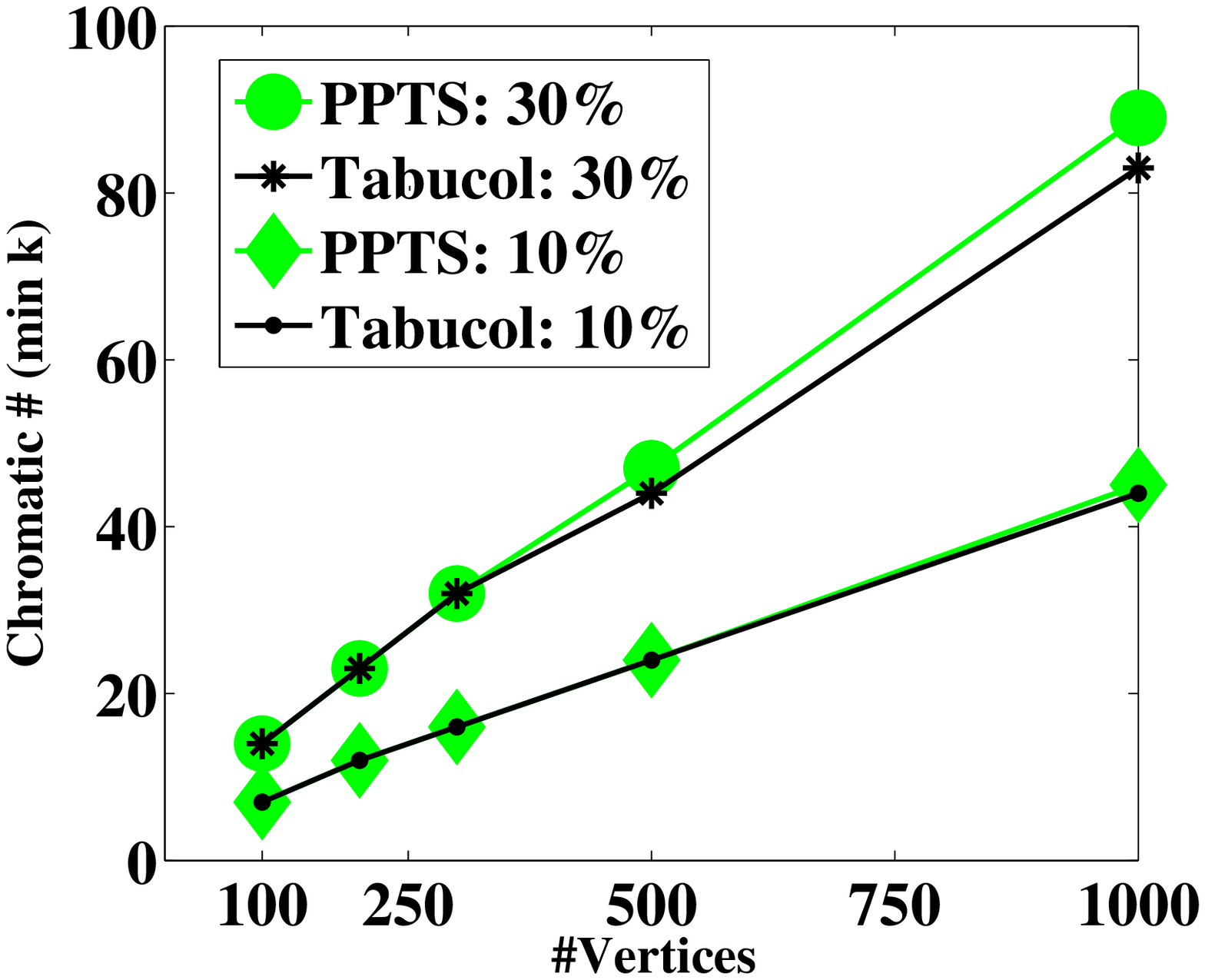}
		\label{fig:k} } \subfigure[Runtime on Graph
	Size]{\centering\includegraphics[angle=0,
		width=0.32\linewidth]{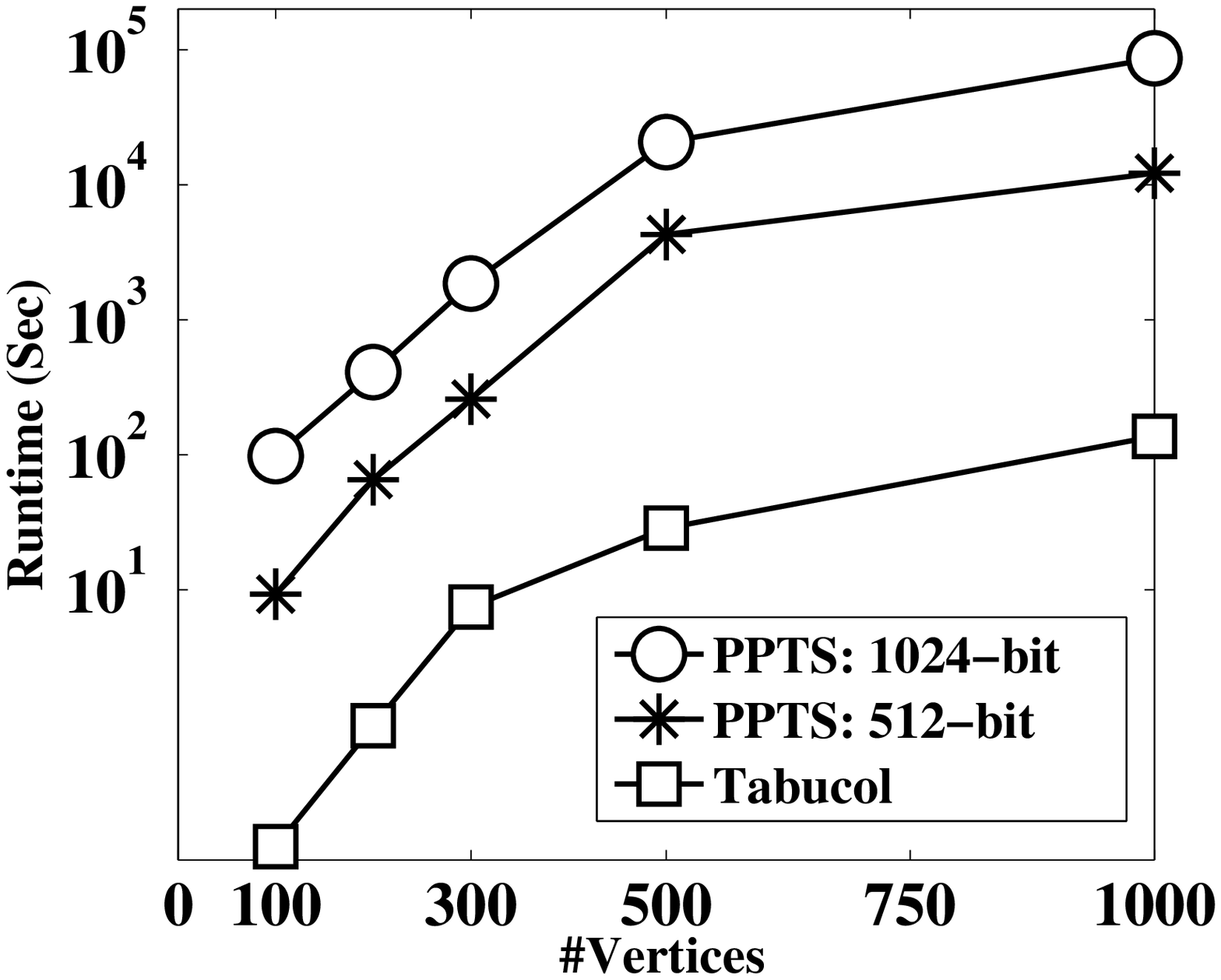} \label{fig:v} } \subfigure[Runtime
	on Density]{\centering\includegraphics[angle=0,
		width=0.32\linewidth]{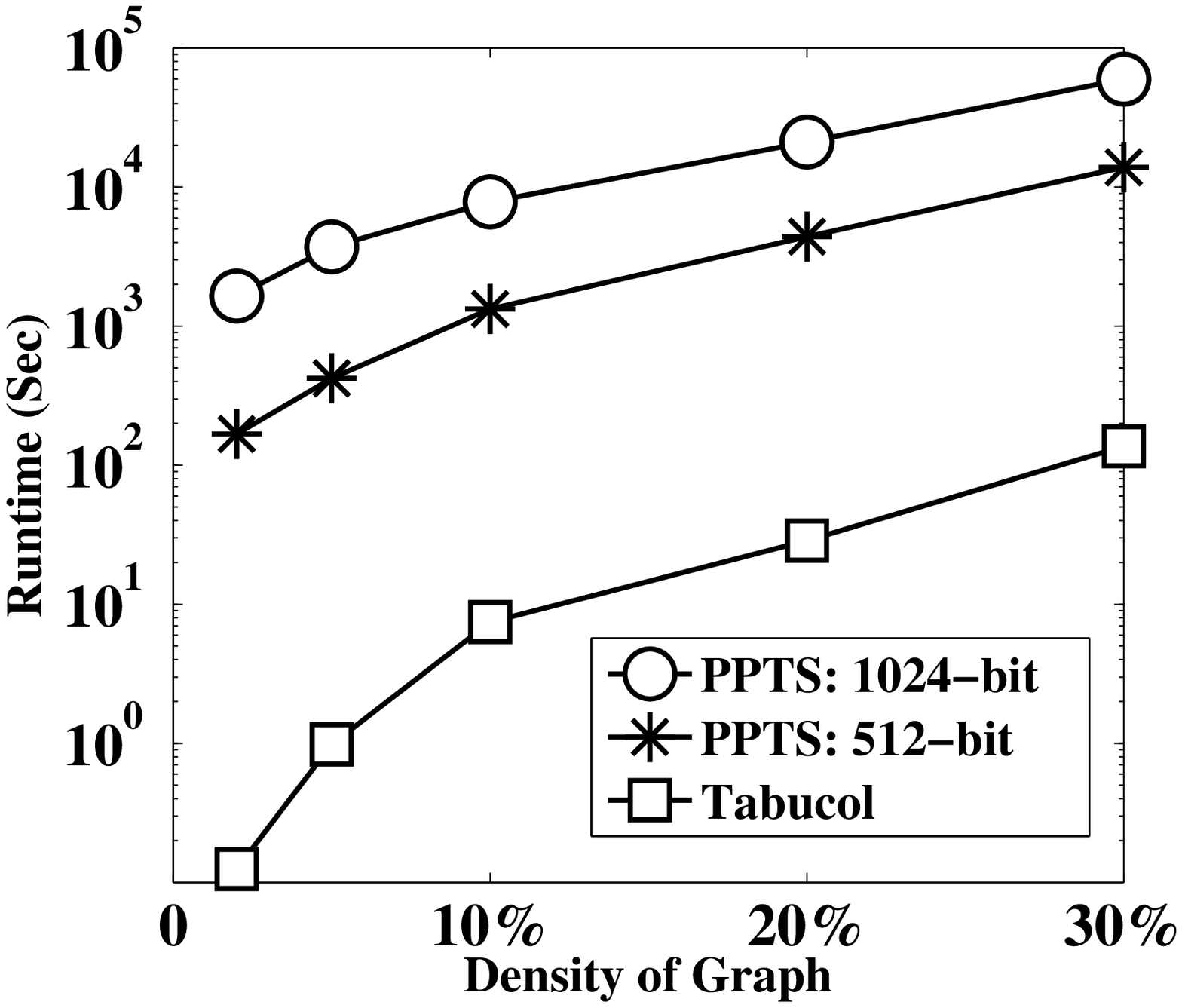} \label{fig:d} } 
	\caption[Example
	of Distributed Combinatorial Optimization] {Optimal Results and
		Computational Cost Testing}  \label{fig:exp}
\end{figure*}

\subsubsection{PPTS Solution Quality}
Fig. \ref{fig:iter} depicts the number of total conflicts $\mu(x)$ and the number of iterations required by PPTS and Tabucol for three different graphs (k is set slightly greater than the optimal chromatic number). While solving
DisGC on these three graphs, the  improvements (reduced $\mu(x)$) in each iteration of both approaches are similar. PPTS generally requires more iterations than Tabucol. This is unavoidable because due to the privacy concern, the algorithm does not know which external edges cause conflict, and therefore cannot improve only upon the conflicting edges. Thus, the search performance is slightly lower, but still acceptable.

Fig. \ref{fig:k} shows the chromatic number (min k) for distributed graphs with different size and density (for each pair of \#V and $\rho$, we randomly generate 10 graphs, solve them with two algorithms and different k, and choose the min k). In general, the optimum performance of the two algorithms is similar.

\subsubsection{PPTS Efficiency}

As the communication cost is small compared to the computation cost (each single computation consumes significantly more time), we evaluate the efficiency of PPTS with respect to computation time. As shown in Fig. \ref{fig:v} and \ref{fig:d}, the runtime of PPTS with 512-bit and 1024-bit is not excessive (approximately 30 hours are required to securely solve a 10-party distributed graph coloring problem with 1000 vertices and density 20\%). This can also be improved by using cryptographic accelerators or parallelizing the computation.

\section{Related Work}
\label{sec:related}
There is prior work both in distributed optimization as well as in secure optimization. Silaghi and Rajeshirke \cite{Silaghi04} show that a secure combinatorial problem solver must necessarily pick the result randomly among optimal solutions to be really secure. Silaghi and Mitra \cite{SilaghiMitra04} propose arithmetic circuits for solving constraint optimization problems that are exponential in the number of variables for any constraint graph. A significantly more efficient optimization protocol specialized on generalized Vickrey auctions and based on dynamic programming is proposed by Suzuki and Yokoo \cite{Suzuki04Vickrey}. Yokoo et al. \cite{Yokoo02DCSP} also propose a scheme using public key encryption for secure distributed constraint satisfaction. Atallah et al.\cite{sec_supplychain} propose protocols for secure supply chain management. However, none of these can be used to solve the general class of distributed scheduling problems, which is what we focus on.

Recently, privacy preserving linear programming problem
\cite{Vaidya09SAC,Vaidya09aina,HongDBsec,HongJCS12,HongOpt13,HongThesis,HongJIS14} and nonlinear programming problem \cite{HongTDSC15} have been well studied to securely solve the distributed LP problems among multiple parties. However, little work has been made on securing distributed NP-hard problems due to the challenging issue of computational complexity. Sakuma et al. \cite{SakumaK07} proposed a genetic algorithm for securely solving two-party distributed traveling salesman problem (NP-hard). They consider the case that one party holds the cost vector while the other party holds the tour vector (this is a special case of multi-party distributed combinatorial optimization). Following their work, the distributed TSP problem that is completely partitioned among multiple parties is securely solved in \cite{HongDBSEC14}. Similarly, we securely solve another classic combinatorial optimization model (graph coloring) that is
completely partitioned among multiple parties. 

Finally, secure computation in real world applications can also be leveraged for the development of privacy preserving optimization. For instance, in the smart grid infrastructure \cite{HongIJER15,icassp18,7929354,SGbook,sgc17} and search engine query applications \cite{HongCikm09,HongEDBT12,HongTDSC15,HongWi11}, the secure communication protocols can be proposed based on the extension of privacy preserving linear or nonlinear programming models (by incorporating more real world constraints).   

\section{Conclusion and Future Work}
\label{sec:con}
In this paper, we have presented a privacy preserving algorithm to solve the distributed graph coloring problem with limited information disclosure. Our solution creates a secure communication protocol based on the existing Tabucol algorithm by making use of several cryptographic building blocks in novel ways. It also provides a roadmap for generalizing the solution to other distributed combinatorial
optimization problems. Given the iterative nature of the solution, we also identify the possible inference attacks on the output in each iteration, and propose a synchronous move mechanism to resolve such problems. In the meanwhile, theoretical study is also given to show the guaranteed protocol security and mitigated inference risks in the proposed approach. We analyze the computation and communication complexity of our modified tabu search, and experimentally show that the PPTS protocol provides nearly identical optimum performance as Tabucol without sacrificing efficiency for privacy.

There are several interesting future directions. While we examine a specific data partitioning for the DisGC problem, other data partitions are also possible. For example, perhaps a party may hold the information about graph edges rather than graph vertices (this could happen, for example in scheduling). We plan to explore this problem and develop solutions for it. Secondly, while graph coloring has many applications, there are several other combinatorial optimization problems, such as mixed integer programming \cite{WangDisasterR18}, scheduling \cite{LiuWHY17,en10010105}, knapsack and vehicle routing \cite{opt} and even some data mining problems \cite{LuHYDB15,HongTDSC12,HongICDMW12,HongVLL16,LuVAH09} which can be naturally distributed. We plan to explore how our privacy preserving tabu search can be applied in these domains. Finally, while tabu-search is generally quite useful, it does not work well for all applications. There are several other algorithms such as simulated annealing, or genetic algorithms that also enjoy wide usage. We will look at how such algorithms could be made privacy preserving in the future.

\section{Acknowledgments}

This work is supported in part by the National Science Foundation under Grant No. CNS-0746943 and the Trustees Research Fellowship Program at Rutgers, the State University of New Jersey.

\appendix

\section{Proof of Lemma 1} \label{sec:lem1}

\begin{proof}
Assuming that $v_i$ is an endpoint of $n$ external edges $e_{ij},j\in[1,n]$, we represent that the colors of $e_{ij}$'s endpoints in $x$ (and $x'$) are conflicted or not as $\mu_{ij}(x)$ = 0 or 1 (and $\mu_{ij}(x')$ = 0 or 1).

If $\mu(x)$ and $\mu(x')$ are revealed in each iteration, while $P_a$ changes the color of $v_i$ (from $x_i$ to $x_i'$), $\sum_{\forall j\in[1,n]}[\mu_{ij}(x')-\mu_{ij}(x)]$ (denoted as $\delta$) is known to $P_a$. Specifically, $\delta\in[-n,n]$ has following cases:

\begin{itemize}
\item \textbf{Case 1: } if $\delta=n$, $P_a$ learns: no edge from $\forall j\in[1,n], e_{ij}$ has a conflict in $x$ , but all the $n$ edges have conflicts in $x'$ (thus, the colors of all $n$ border vertices from other parties are
identical to $x_i'$);

\item \textbf{Case 2: } if $\delta=-n$, $P_a$ learns: all the edges have conflicts in $x$ but all the conflicts are eliminated in $x'$ (thus, the colors of all $n$ border vertices from other parties are identical to $x_i$);

\item \textbf{Case 3: } more generally, if $-n<\delta<n$, every edge has a probability to have conflict. We denote the probability that a specific edge $e_{ij}$ has no conflict in $x$ but has conflict in $x'$ (and has
conflict in $x$ but has no conflict in $x'$) as $Prob[\mu_{ij}(x\rightarrow x'):0\rightarrow 1]$ (and
$Prob[\mu_{ij}(x\rightarrow x'):1\rightarrow 0]$) where $\delta\in(-n,n)$ (since inferring that $e_{ij}$ has conflict in either $x$ or $x'$ should result in learning $v_j$'s color by $P_a$).
\end{itemize}

$\diamond.$ We first consider that $\delta\geq 0$ (if $\delta<0$, we can obtain a similar result). We denote the number of edges (out of $n$) where $\mu_{ij}(x)\rightarrow\mu_{ij}(x'):0\rightarrow 1$ as $\delta_1$; denote the number of edges (out of $n$) where $\mu_{ij}(x)\rightarrow\mu_{ij}(x'):1\rightarrow 0$ as $\delta_0$;
denote the number of edges (out of $n$) where $\mu_{ij}(x)\rightarrow\mu_{ij}(x'):$ $0\rightarrow 0$ as $\delta_c$ ($1\rightarrow 1$ is impossible since $x_i$ is changed). Thus, we have $\delta_1+\delta_0+\delta_c=n$. Since $\delta_1-\delta_0=\delta$ (overall difference), we have
$\delta_1\leq\lfloor\frac{n+\delta}{2}\rfloor$ and $\delta_0\leq\lfloor\frac{n-\delta}{2}\rfloor$ (if $n-\delta$ is an
odd number, due to $\delta_c\geq0$, we have $\delta_1<\frac{n+\delta}{2}$). Furthermore, it is clear that
$\delta_1\geq \delta$. We now discuss all possible values of $(\delta_1,\delta_0)$.

Specifically, given $\delta$ and $n$, $(\delta_1,\delta_0)$ has following possible pairs of values: $(\delta,0)$,
$(\delta+1,1),\dots$, $(\lfloor\frac{n+\delta}{2}\rfloor$, $\lfloor\frac{n-\delta}{2}\rfloor)$ (equally possible for all
pairs). Hence, the probability of valuing each pair is $\frac{1}{\lfloor\frac{n-\delta}{2}\rfloor+1}$. As a consequence, for a given edge, $Prob[\mu_{ij}(x\rightarrow x'):0\rightarrow 1]=\frac{1}{\lfloor\frac{n-\delta}{2}\rfloor+1}\times[\frac{\delta}{n}+\frac{\delta+1}{n}+\dots+\frac{\lfloor\frac{n+\delta}{2}\rfloor}{n}]$=$\frac{\delta+\lfloor\frac{n+\delta}{2}\rfloor}{2n}$ and $Prob[\mu_{ij}(x\rightarrow x'):1\rightarrow 0]=\frac{1}{\lfloor\frac{n-\delta}{2}\rfloor+1}\times[\frac{0}{n}+\frac{1}{n}+\dots+\frac{\lfloor\frac{n-\delta}{2}\rfloor}{n}]$=$\frac{\lfloor\frac{n-\delta}{2}\rfloor}{2n}$. We thus have:

\begin{enumerate}
\item since the probability $Prob[\mu_{ij}(x\rightarrow x'):0\rightarrow 1]=Prob[x_j=x_i']$ with a known fixed $\delta$, party $P_a$ can guess that $x_j=x_i'$ and the endpoint colors of $e_{ij}$ are conflicted in solution $x'$ with a probability of $\frac{\delta+\lfloor\frac{n+\delta}{2}\rfloor}{2n}$.

\item since the probability $Prob[\mu_{ij}(x\rightarrow x'):1\rightarrow 0]=Prob[x_j=x_i]$ with a known fixed $\delta$, party $P_a$ can guess that $x_j=x_i$ and the endpoint colors of $e_{ij}$ are conflicted in solution $x$ with a probability of $\frac{\lfloor\frac{n-\delta}{2}\rfloor}{2n}$.
\end{enumerate}

$\diamond.$ Similarly, if $\delta<0$, we can deduce that $Prob[x_j=x_i]=Prob[\mu_{ij}(x\rightarrow x'):1\rightarrow 0]= \frac{-\delta+\lfloor\frac{n-\delta}{2}\rfloor}{2n}$ and $Prob[x_j=x_i']=Prob[\mu_{ij}(x\rightarrow x'):0\rightarrow 1]= \frac{\lfloor\frac{n+\delta}{2}\rfloor}{2n}$.

Finally, we verify Case 1 and 2 in the general probability computation: if $\delta=n$, we have $Prob[\mu_{ij}(x\rightarrow x'):0\rightarrow 1]=1$ and $Prob[\mu_{ij}(x\rightarrow x'):1\rightarrow 0]=0$, hence the color of $v_j$ is determined as $x_i'$; if $\delta=-n$, we have $Prob[\mu_{ij}(x\rightarrow
x'):0\rightarrow 1]=0$ and $Prob[\mu_{ij}(x\rightarrow x'):1\rightarrow 0]=1$, hence the color of $v_j$ is determined as $x_i$.

\end{proof}

\section{Proof of Lemma 2}
\label{sec:lem2}
\begin{proof}
We let $\mu_a(x)$ and $\mu_a(x')$ be the number of $P_a$'s conflicted internal edges in $x$ and $x'$ respectively, and denote $\delta_a$ as $\mu_a(x')-\mu_a(x)$. If only $\mu(x')<\mu(x')$ or $\mu(x')\geq \mu(x)$ are revealed in each iteration, while $P_a$ changes the color of $v_i$ (from $x_i$ to $x_i'$), $P_a$ knows only
$\delta+\delta_a<0$ or $\delta+\delta_a\geq0$ ($\delta=\sum_{\forall j\in[1,n]}(\mu_{ij}(x')-\mu_{ij}(x))$). We now try to compute the probability $Prob[x_j=x']$ and $Prob[x_j=x]$. Essentially, if $\delta$ is known and fixed, $Prob[x_j=x']$ and $Prob[x_j=x]$ can be computed in terms of Lemma 1. Hence, we discuss the possible privacy
breaches for all cases. Since $P_a$ knows $\delta_a$ and the relationship between $\delta+\delta_a$ and $0$, we look at two cases:

\begin{itemize}

\item \textbf{Case 1: } if $\delta+\delta_a<0$, we have $\delta<-\delta_a$. Since $\delta_a$ is known to $P_a$ and it can be any integer, we discuss all possible $\delta_a$ as following.

\begin{enumerate}

\item if $\delta_a\geq n$, it is impossible to obtain $\delta+\delta_a<0$ because $\delta\in[-n,n]$;

\item if $\delta_a<-n$, $\delta$ can be any integer in $[-n,n]$ to generate output $\mu(x')< \mu(x)$, hence, probabilities $Prob[x_j=x_i']$ and $Prob[x_j=x_i]$ is minimized (no additional information is revealed);

\comment{ $\forall\delta\in[-n,n]$, $\delta+\delta_a<0$ holds. Therefore, $Prob[x_j=x_i']=\frac{1}{2n+1}\times \sum_{\forall \delta\in [0,n]}Prob[x_j=x_i'|\delta]+\frac{1}{2n+1}\times\sum_{\forall
\delta\in [-n,-1]} Prob[x_j=x_i'|\delta]= \frac{1}{4}$ (Sketch Derivation). Similarly, $Prob[x_j=x_i]=\frac{1}{4}<1$;}

\item if $-n\leq\delta_a\leq n-1$, $P_a$ can learn that $-n\leq\delta<-\delta_a$ (regularly, $-n\leq\delta\leq n$ is known to $P_a$). With the greater $-\delta_a$, $P_a$ learns less additional information. The worst privacy breach case should be $\delta_a=n-1$ such that $\delta=-n$ is inferred. At this time, $P_a$ can infer $x_j$ with probability $Prob[x_j=x_i]=1$.

\end{enumerate}

\item \textbf{Case 2: } if $\delta+\delta_a\geq0$, we have $\delta\geq-\delta_a$. Similarly, we discuss following cases:

\begin{enumerate}

\item if $\delta_a< -n$, it is impossible to obtain $\delta+\delta_a\geq0$ because $\delta\in[-n,n]$;

\item if $\delta_a\geq n$, $\delta$ can be any integer in $[-n,n]$ to generate output $\mu(x')\geq \mu(x)$,
hence, probabilities $Prob[x_j=x_i']$ and $Prob[x_j=x_i]$ is minimized (no additional information is revealed);

\item if $-n\leq\delta_a<n$, $P_a$ can learn that $-\delta_a\leq\delta\leq n$ (regularly, $-n\leq\delta\leq n$ is
known to $P_a$). With the smaller $-\delta_a$, $P_a$ learns less additional information. The worst privacy breach case should be $\delta_a=-n$ such that $\delta=n$ is inferred. At this time, $P_a$ can infer $x_j$ with probability $Prob[x_j=x_i']=1$.

\end{enumerate}

\end{itemize}

In sum, with greater $-\delta_a$ in Case 1 and smaller $-\delta_a$ in Case 2, it should be more difficult to guess $\delta$ and infer $x_j=x_i$ or $x_j=x_i'$ in most cases, compared with known $\delta$ (to $P_a$). Hence, the risk of inference attack has been reduced in general. However, we have two worst cases that still breaches
privacy: 1. As $\delta_a=n-1$ and $\mu(x')<\mu(x)$ are the output to $P_a$ in an iteration, thus $P_a$ learns $\forall j\in[1,n],x_j=x_i$; 2. As $\delta_a=-n$ and $\mu(x')\geq\mu(x)$ are the output to $P_a$ in an iteration, thus $P_a$ learns $\forall j\in[1,n], x_j=x_i'$.

\end{proof}

\section{Proof of Lemma 3}
\label{sec:lem3}

\begin{proof}
\begin{table}[h!]
\centering \caption{Notations in Synchronous Move}
\begin{tabular}{|c|c|}
      \hline
      $v_i$ and $v_j$ & $P_a$ and $P_b$'s vertex for move\\
      \hline
      $\mu_a(x)$ and $\mu_a(x')$ & $P_a$'s number of conflicted internal edges in $x$ and $x'$ respectively\\
      \hline
      $\mu_b(x)$ and $\mu_b(x')$ & $P_b$'s number of conflicted internal edges in $x$ and $x'$ respectively\\
      \hline
      $\delta_a$ and $\delta_b$ & $\delta_a=\mu_a(x')-\mu_a(x)$, $\delta_b=\mu_b(x')-\mu_b(x)$\\
      \hline
      $\forall s, e_{is}$; $\forall t, e_{jt}$ &
      $n$ external edges with endpoint $v_i$; $n'$ external edges with endpoint $v_j$\\
      \hline
      $\forall s, \mu_{is}(x)$ and $\mu_{is}(x')$ & $e_{is}$ is conflicted or not in $x$ and
      $x'$\\
      \hline
      $\forall t, \mu_{jt}(x)$ and $\mu_{jt}(x')$& $e_{jt}$ is conflicted or not in $x$ and
      $x'$\\
      \hline
      $\delta$ & $\delta=\sum_{\forall s}[\mu_{is}(x')-\mu_{is}(x)]$\\
      \hline
      $\delta'$& $\delta'=\sum_{\forall t}[\mu_{jt}(x')-\mu_{jt}(x)]$\\
      \hline
\end{tabular}

\label{table:notation}
\end{table}
\vspace{-0.1in}

We first present some notations in synchronous move in Table \ref{table:notation}. With a synchronous move, $\mu(x')<\mu(x)$ or $\mu(x')\geq \mu(x)$ is revealed. $\delta_a$ and $\delta_b$ are known to $P_i$ and $P_j$ respectively. We now discuss the potential inferences on the above revealed information.

First, if $\mu(x')<\mu(x)$, we thus have $\delta_a+\delta+\delta_b+\delta'<0$. $P_a$ knows $\delta_a$ and
$\delta\in[-n,n]$. Since $\delta_b$ and $\delta'$ can be any integer, no additional information other than
$\delta<-\delta_a-\delta_b-\delta'$ can be inferred (as soon as $\delta_a$ is very large, $\delta$ can still be any integer in $[-n,n]$ because $\delta_b+\delta'$ may be negative and sufficiently small to anonymize $\delta$). Similarly, if $v_j$ is a border vertex of $P_b$, $P_b$ cannot learn any additional information other than
$\delta_a+\delta+\delta_b+\delta'<0$ and $\mu(x')<\mu(x)$.

Second, if $\mu(x')\geq\mu(x)$, we thus have $\delta_a+\delta+\delta_b+\delta'\geq0$. For the same reason,
$\delta$ is an unknown integer between $[-n,n]$ to $P_i$, and $\delta'$ is an unknown integer between $[-n',n']$ to $P_j$. No additional information can be inferred.

In a special case, $P_a$'s $v_i$ is adjacent to a vertex of $P_b$ (this is unknown to $P_b$). $P_b$ knows that $P_a$ is moving a border vertex and the result of securely comparing $\mu(x)$ and $\mu(x')$. However, $P_b$ cannot infer that $P_a$ is moving a border vertex that is adjacent to $P_b$ without colluding with other
parties (since $P_a$ may have many border vertices, adjacent to many parties). Hence, synchronous move guarantees security for this special case.

Thus, the inference attacks stated in Lemma 1 and 2 can be eliminated.
\end{proof}

\end{document}